\documentclass{article}
\usepackage[left=1in, right=1in, top=1in, bottom=1in]{geometry}
\usepackage{enumitem}
\usepackage{authblk}
\usepackage{graphicx}	
\usepackage{amsmath, amsthm, amssymb}
\usepackage{xspace}
\usepackage{comment} 
\usepackage{hyperref}
\usepackage{xcolor}
\usepackage{wrapfig}
\usepackage{tikz}
\usepackage{tablefootnote}
\usetikzlibrary{positioning}
\usetikzlibrary{arrows}
\usetikzlibrary{decorations.pathmorphing}
\usetikzlibrary{decorations.markings}
\usepackage{tikz-cd}
\usepackage{cleveref}
\usepackage{booktabs}
\usepackage{float}

\newcommand{\algprobm}[1]{\textsc{#1}\xspace}

\newcommand{\Soc}{\text{Soc}}

\newcommand{\aut}{\mathrm{Aut}}

\theoremstyle{plain}
\newtheorem{theorem}{Theorem}[section]
\newtheorem{proposition}[theorem]{Proposition}
\newtheorem{corollary}[theorem]{Corollary}
\newtheorem{lemma}[theorem]{Lemma}

\newtheorem{observation}[theorem]{Observation}

\theoremstyle{definition}
\newtheorem{definition}[theorem]{Definition}
\newtheorem{remark}[theorem]{Remark}

\newtheorem{problem}[theorem]{Problem}

\newcommand{\betacc}[1]{\ifthenelse{\equal{#1}{1}}{\exists^{\log n}}{\exists^{\log^{#1}n}}} 
\newcommand{\alphacc}[1]{\ifthenelse{\equal{#1}{1}}{\forall^{\log n}}{\forall^{\log^{#1}n}}} 

\DeclareMathOperator{\GL}{GL}
\DeclareMathOperator{\SL}{SL}
\DeclareMathOperator{\PSL}{PSL}
\newcommand{\PSp}{\mathrm{PSp}}

\newcommand{\POmegaPlus}{\mathrm{P}\Omega^{+}}

\newcommand{\POmega}{\mathrm{P}\Omega}

\DeclareMathOperator{\Aut}{Aut}

\DeclareMathOperator{\rad}{Rad}

\DeclareMathOperator{\poly}{poly}

\newcommand{\ord}{\mathrm{ord}}
\newcommand{\G}{G_2}
\newcommand{\F}{F_4}
\newcommand{\Esix}{E_6}

\newcommand{\PGammaL}{{\mathrm{P\Gamma L}}(d,q)}
\newcommand{\mathbbm}[1]{\mathfrak{#1}}

\newcommand*{\Sym}{\text{Sym}}

\usepackage{lineno}

\title{Complexity of Constructing Minimal Faithful Permutation Representations for Fitting-free Groups\footnote{A preliminary version of this work appeared in the proceedings of FCT 2025 \cite{LevetSrivastavaThakkarFCT}.
}}
\author[1]{Michael Levet}
\author[2]{Pranjal Srivastava}
\author[3]{Dhara Thakkar}
\affil[1]{Department of Computer Science, College of Charleston}
\affil[2]{Department of Math and Computing, Indian Institute of Information Technology Vadodara, India}
\affil[3]{Graduate School of Mathematics, Nagoya University, Japan}

\begin{document}
\maketitle

\begin{abstract}
In this paper, we investigate the complexity of computing 
minimal faithful permutation representations for groups without abelian normal subgroups (a.k.a. Fitting-free groups). When our groups are given as quotients of permutation groups, we exhibit a polynomial-time algorithm for constructing such representations. Furthermore, in the setting of permutation groups, we obtain an $\textsf{NC}$ procedure for computing the minimal faithful permutation degree, and a randomized $\textsf{NC}$ ($\textsf{RNC}$) algorithm for computing a minimal faithful permutation representation. This improves upon the work of Das and Thakkar (STOC 2024, \textit{SIAM J. Comput.} 2026), who established a Las Vegas polynomial-time algorithm for computing the minimal faithful permutation degree for this class in the setting of permutation groups.
\end{abstract}

\newpage
\section{Introduction}

Cayley's theorem states that for every finite group $G$, there exists $m \in \mathbb{N}$ such that $G$ can be embedded into $\Sym(m)$. The \emph{minimal faithful permutation degree}, denoted $\mu(G)$, is the smallest such $m$ for which $G$ can be embedded into $\Sym(m)$. A permutation representation of $G$ on $\mu(G)$ points is called a \emph{minimal faithful permutation representation}. Note that $\mu(G) \leq |G| \leq \mu(G)!$. By convention, $\mu(G) = 0$ if and only if $G = 1$. Johnson \cite{JohnsonMPD} exhibited the following characterization of $\mu(G)$.

\begin{proposition}[{\cite{JohnsonMPD}}] \label{prop:Johnson}
Let $G$ be a finite group, and let $\mathcal{G}$ be the collection of all subgroups of $G$. We have that:
\[
\mu(G) := \min_{\emptyset \neq \mathcal{H} \subseteq \mathcal{G}} \left\{ \sum_{H \in \mathcal{H}} [G : H] \ \biggr| \bigcap_{H \in \mathcal{H}} \text{Core}_{G}(H) = 1 \right\},
\]
where $\text{Core}_{G}(H) = \bigcap_{g \in G} gHg^{-1}$ and $|\mathcal{H}| \leq \log |G|$.\footnote{The fact that we can require any witnessing collection of subgroups $\mathcal{H}$ to satisfy $|\mathcal{H}| \leq \log |G|$ was carefully analyzed by \cite{DasThakkarMPD}.}
\end{proposition}

In this paper, we investigate the computational complexity of computing $\mu(G)$. Precisely, we will consider the decision variant, \algprobm{Min-Per-Deg}, which takes as input a finite group $G$ and $k \in \mathbb{N}$, and asks if $\mu(G) \leq k$. We will also consider the search variant of \algprobm{Min-Per-Deg}, for which we will construct an explicit embedding $\varphi : G \to \Sym(\mu(G))$. Key motivation for this comes from the setting of permutation group algorithms. Given a group $G \leq \text{Sym}(n)$, there is a natural divide-and-conquer strategy where we decompose $G$ in terms of a normal subgroup $N \trianglelefteq G$ and the corresponding quotient $G/N$. We then seek to solve the algorithmic problem on $N$ and $G/N$, before recombining the solutions on the subproblems to a solution for $G$. The techniques involved when reasoning about quotients of two permutation groups are considerably more involved than when dealing with a single permutation group \cite{KantorLuksQuotients}. Thus, if $G/N$ has a \textit{small} faithful permutation representation that can be efficiently constructed, then we can take advantage of the extensive suite of permutation group algorithms (e.g., \cite{CH03, SeressBook, BabaiLuksSeress, LuksReduction}). This approach is not always possible, as there exist quotients with exponentially large degree. Let $D_{8}$ be the dihedral group of order $8$. Neumann observed that $D_{8}^{\ell}$ admits a faithful permutation representation on $4\ell$ points, but admits a quotient $Q$ that is an extraspecial $2$-group with $\mu(Q) \geq 2^{\ell+1}$ \cite{neumann1985}. For a survey of algorithms to deal with quotient groups, we refer to \cite{KantorLuksQuotients, WilsonDirectProductsArxiv, LuksMiyazaki, SeressBook}.

The parameter $\mu(G)$ has received considerable attention, primarily from a mathematical perspective. Johnson showed that the Cayley representation of $G$ (that is, the representation of $G$ induced by the left-regular action of $G$ on itself) is minimal ($\mu(G) = |G|$) if and only if $G$ is cyclic of prime power order, a generalized quaternion $2$-group, or the Klein group of order $4$ \cite{JohnsonMPD}. Povzner \cite{Povzner} and Johnson \cite{JohnsonMPD} independently showed that for an abelian group $G$, $\mu(G)$ is the sum of the prime powers that occur in the direct product decomposition of $G$ into indecomposable factors. For a general survey, we refer to \cite{wright1974,neumann1985,easdownPraeger1988,Holt1997,babai1993,kovacs2000,elias2010,HoltWalton2002,Saunders2011} and the references therein.

 To the best of our knowledge, Kantor and Luks \cite[Sec.~13, Problem 1]{KantorLuksQuotients} were the first to propose investigating the computational complexity of \algprobm{Min-Per-Deg}. In the setting of permutation groups and their quotients, \algprobm{Min-Per-Deg} belongs to $\textsf{NP}$ \cite{KantorLuksQuotients, DasThakkarMPD}.
 Proposition~\ref{prop:Johnson} yields a trivial upper-bound of $|G|^{O(\log^{2} |G|)}$-time for this problem, in the multiplication (Cayley) table model. This is essentially the generator-enumeration strategy. Namely, we enumerate over all collections $\mathcal{H}$ of subgroups of $G$, where $|\mathcal{H}| \leq \log |G|$. Each subgroup in $\mathcal{H}$ is specified by generators, and so we require at most $\log |G|$ generators per subgroup in $\mathcal{H}$. Note that, given generators for a subgroup $H \leq G$, we may write down the elements of $H$ in polynomial time. Given the explicit elements belonging to each subgroup $H \in \mathcal{H}$, we may easily handle the remainder of the computations. 

Das and Thakkar \cite{DasThakkarMPD}  recently made significant progress in the computational complexity of the decision variant of \algprobm{Min-Per-Deg}. Note that their work was non-constructive; they computed the minimal faithful permutation degree $\mu(G)$, but did not construct witnessing representations $\varphi : G \to \Sym(\mu(G))$. In the setting of permutation groups, they exhibited a Las Vegas polynomial-time algorithm for groups without abelian normal subgroups, and a quasipolynomial-time algorithm for primitive permutation groups. When the groups are given by their multiplication (Cayley) tables, Das and Thakkar exhibited an upper bound of $\textsf{DSPACE}(\log^{3} |G|)$, which improves upon the trivial $|G|^{O(\log^2 |G|)}$ bound. 

\noindent \\ \textbf{Our Contributions.} In this paper, we further investigate both the decision and search variants of the \algprobm{Min-Per-Deg} problem for groups without abelian normal subgroups, addressing one of the open questions posed by Das and Thakkar \cite{DasThakkarMPD}. Our main result is the following.

\begin{theorem} \label{thm:MainFittingFree}
Let $\textbf{G}$ be a group without abelian normal subgroups (a.k.a. Fitting-free).\footnote{In this paper, we will use the term \textit{Fitting-free} to refer to the class of groups without abelian normal subgroups, and \textit{semisimple} to refer to algebraic structure decomposes as a direct sum of simple objects. The terms \textit{Fitting-free} or \textit{trivial-Fitting} are common in the group theory literature (see e.g., \cite{CH03, Holt2005HandbookOC}), while the term \textit{semisimple} has been used in the complexity theory literature to refer to the class of groups without abelian normal subgroups (see e.g., \cite{BCGQ, BCQ}).}
\begin{enumerate}[label=(\alph*)]
\item Suppose $\textbf{G}$ is given as the quotient $G/K$ of two permutation groups $K \trianglelefteq G \leq \text{Sym}(n)$. Then we can compute $\mu(\textbf{G})$, as well as a minimal faithful permutation representation $\varphi : \textbf{G} \to \Sym(\mu(G))$, in polynomial time.

\item If $\textbf{G}$ is given as a permutation group ($K = 1$), then we can compute $\mu(\textbf{G})$ in $\textsf{NC}$. Furthermore, we can construct a minimal faithful permutation representation $\varphi : \textbf{G} \to \Sym(\mu(G))$ using a randomized $\textsf{NC}$ ($\textsf{RNC}$) algorithm.
\end{enumerate}
\end{theorem}

Given a finite group $G$ by its multiplication (Cayley) table, we may easily construct the regular representation in $\textsf{NC}$. Thus, Theorem~\ref{thm:MainFittingFree} yields the following immediate corollary.

\begin{corollary}
Let $G$ be a Fitting-free group given by its multiplication table. We can compute $\mu(G)$ in $\textsf{NC}$, as well as a minimal faithful permutation representation $\varphi : \textbf{G} \to \Sym(\mu(G))$, in $\textsf{RNC}$.    
\end{corollary}

We compare Theorem~\ref{thm:MainFittingFree} to the previous result of Das and Thakkar \cite{DasThakkarMPD}. Here, Das and Thakkar established a Las Vegas polynomial-time analogue of Theorem~\ref{thm:MainFittingFree} in the setting of \textit{permutation groups}. They explicitly did not deal with quotients \cite{DasThakkarMPD}. Das and Thakkar crucially leveraged randomness by using a suite of Las Vegas algorithms to constructively recognize finite simple groups in the black-box model \cite{KantorSeress}. Note that the black-box model also includes matrix groups, where order-finding is already at least as hard as \algprobm{Discrete Logarithm}, and so Las Vegas polynomial-time is essentially optimal in the black-box setting.

In a preliminary version of this work \cite{LevetSrivastavaThakkarFCT}, we investigated the decision variant of \algprobm{Min-Per-Deg} for Fitting-free groups. Here, we showed that for a Fitting-free group $\textbf{G}$, we can compute $\mu(\textbf{G})$ in (i) polynomial-time if $\textbf{G}$ is specified as the quotient of two permutation groups, and (ii) $\textsf{NC}$ if $\textbf{G}$ is specified as a permutation group. In this work, we further extend \cite{LevetSrivastavaThakkarFCT} by constructing explicit permutation representations.

\noindent \\ \textbf{Techniques for computing $\mu(G)$.} We will first discuss our methods for handling the decision variant of \algprobm{Min-Per-Deg} (computing $\mu(G)$). In order to compute $\mu(G)$, we take advantage of several previous works, including the extensive suite of polynomial-time algorithms for quotients \cite{KantorLuksQuotients} and $\textsf{NC}$ algorithms for permutation groups \cite{BabaiLuksSeress}, as well as standard parallel algorithms for linear algebra (e.g., \cite{Eberly91, MulmuleyRank}). Additionally, we crucially leverage the framework of Kantor, Luks, and Mark \cite{KantorLuksMark, MarkThesis} for constructive recognition of finite simple groups in $\textsf{NC}$ (see Section~\ref{sec:ConstructiveRecognition} for more details).

Still, in order to compute $\mu(\textbf{G})$ in the setting of permutation groups, we require an additional ingredient; namely, computing the socle\footnote{Recall that the \emph{socle} of a group $G$ is the direct product of the minimal normal subgroups of $G$.} in $\textsf{NC}$. While such polynomial-time algorithms have long been known for both permutation groups (see e.g., \cite{SeressBook, Holt2005HandbookOC}) and their quotients \cite{KantorLuksQuotients}, it is open whether we can compute the socle in $\textsf{NC}$. Prior to this work, it was open, even in the restricted setting where we have a permutation group that is Fitting-free. To this end, we establish the following.

\begin{proposition} \label{prop:MainSocle}
Let $G$ be a Fitting-free group specified by a generating sequence of permutations from $\Sym(n)$. We can compute $\Soc(G)$ in $\textsf{NC}$. 
\end{proposition}

In order to prove Proposition~\ref{prop:MainSocle}, we take advantage of an $\textsf{NC}$-computable series of length $O(\log^2 n)$, where each term is normal in $G$ and the successive quotients are semisimple (direct products of simple groups) \cite{BabaiLuksSeress}. For an arbitrary group, $\Soc(G)$ may contain abelian normal subgroups. Constructing generators for both the abelian and non-abelian parts of the socle together, in $\textsf{NC}$, remains open. This was an obstacle to computing $\Soc(G)$ in polynomial time, in the setting of permutation groups \cite{BabaiKantorLuksCFSG}.

Extending Proposition~\ref{prop:MainSocle} to the setting of quotients is also non-trivial (see Remark~\ref{rmk:Quotients} for further discussion). To prove Proposition~\ref{prop:MainSocle}, we take advantage of the fact that we can, in $\textsf{NC}$, compute the centralizer of a normal subgroup \cite{BabaiLuksSeress}. While such a polynomial-time algorithm exists for quotients, the machinery is considerably more complicated than for permutation groups and appears resistant to parallelization \cite[Section~6]{KantorLuksQuotients}. This, in turn, presents a significant obstacle to improving the bound in Theorem~\ref{thm:MainFittingFree}(a) from $\textsf{P}$ to $\textsf{NC}$. \\

\noindent \textbf{Techniques for Computing Minimal Faithful Permutation Representation.} While our work on computing $\mu(\textbf{G})$, building on the work of Das and Thakkar \cite{DasThakkarMPD}, serves as a foundation, we still require several additional ingredients to compute a minimal faithful permutation representation $\varphi : \textbf{G} \to \Sym(\mu(\textbf{G}))$. Recall that the \emph{socle} $\Soc(\textbf{G})$ is the direct product of the minimal normal subgroups of $\textbf{G}$. As $\textbf{G}$ is Fitting-free, each minimal normal subgroup of $\textbf{G}$ is of the form $N = S^k$, for some non-abelian simple group $S$. Let $S_1$ be one such non-abelian simple direct factor of $N$. Let $A \leq \Aut(S_1)$ be the almost-simple subgroup induced by the conjugation action of $N_{\textbf{G}}(S_1)$ on $S_1$. Our first step is to build a minimal faithful permutation representation of $A$. For computing $\mu(\textbf{G})$, we only require constructive recognition of finite simple groups. However, in order to compute an embedding $\varphi : A \to \Sym(\mu(A))$, we require minimal faithful permutation representations for each non-abelian simple group $S$. 

If $|S| \in \poly(n)$, then we can write down the multiplication table for $S$. As $S$ is simple, we have from \cite{JohnsonMPD} that there exists a maximal subgroup $P \leq S$ such that the regular action of $S$ on the right cosets of $S/P$ is a minimal faithful permutation representation. Kantor \cite{kantorPrimeOrderElement} established that if $S$ is an exceptional group of Lie type, then $|S| < n^9$ (recalled as Theorem~\ref{thm:Kantor}). Note that there are $26$ sporadic groups, and so these all have bounded order. Thus, for the exceptional groups of Lie type and the sporadic groups, we may construct a minimal faithful permutation degree in this way (see Lemma~\ref{Minper-Simple}).

In the case when $|S| \geq n^9$, then $S$ is either an alternating group or a classical simple group. We do not require detailed knowledge of the minimal faithful permutation representation of the alternating groups. However, we crucially utilize the fact that, except in a few small cases, a classical simple group of dimension $m$ over $\mathbb{F}_{q}$ admits a minimal faithful permutation representation over an appropriate orbit of $1$-spaces from $\mathbb{F}_{q}^{m}$. A careful reading of the work of Kantor, Luks, and Mark \cite{KantorLuksMark} (see in particular, Mark's thesis \cite{MarkThesis}) shows that if $S$ is a classical simple group of order $|S| \geq n^9$, that their procedure for constructive recognition of finite simple groups returns a minimal faithful permutation representation detailing the action of $S$ on the corresponding orbit of $1$-spaces. While Kantor, Luks, and Mark need not construct such representations for classical simple groups of order less than $n^9$, we will need to do so (Lemma~\ref{lem:MinPerClassicalSimple}).

Now if $S$ is a classical or exceptional simple group of Lie type (excluding $\POmegaPlus(8,q)$), then every automorphism of $S$ can be decomposed in terms of inner, diagonal, field, and graph automorphisms \cite{carter-book} (see Section~\ref{subsec: simplegroup} for more details). In the case when $S = \POmegaPlus(8,q)$, we have the inner, diagonal, and field automorphisms, as well as the triality automorphism instead of the graph automorphism \cite{HallTrialityNotes}. For these cases, we will need to build permutations for the inner, field, diagonal, and graph and triality automorphisms. The key technical challenge is to build permutation representations for the graph and triality automorphisms in the cases where $S$ is a classical simple group. While classical simple groups act naturally on orbits of $1$-spaces, the graph and triality automorphisms introduce new auxiliary vector spaces. We leverage the Classification of Finite Simple Groups (see e.g., \cite{ref4MV, ref7VClassical, carter-book, cameronnotes, HallTrialityNotes}) detailing the actions of the graph and triality automorphisms, in order to build the representative permutations.

Suppose now that we have constructed a minimal faithful permutation representation $\varphi_{i}$ for each almost-simple group $A_1, \ldots, A_k$, corresponding to each minimal normal subgroup $N_1, \ldots, N_k$ of $\textbf{G}$. It remains to coalesce $\varphi_{1}, \ldots, \varphi_{k}$ into a minimal faithful permutation representation of $\textbf{G}$. We accomplish this using the techniques of Cannon and Holt \cite{CH04}. One key technicality here is in constructing a transversal for the cosets of a subgroup $\textbf{H} \leq \textbf{G}$ of polynomial index. While such a transversal can be constructed in polynomial time for both permutation groups and their quotients \cite{LuksReduction}, it is a long-standing open problem as to whether this is doable in $\textsf{NC}$ for permutation groups \cite{BabaiLuksSeress}. However, Babai, Luks, and Seress \cite{BabaiLuksSeress} observed that the techniques of Babai \cite{BabaiMonteCarlo} yield a \emph{randomized} $\textsf{NC}$ algorithm ($\textsf{RNC}$) for constructing such a transversal for permutation groups. This is the sole obstacle for obtaining $\textsf{NC}$ bounds to construct a minimal faithful permutation representation for Theorem~\ref{thm:MainFittingFree}(b).

We note that Kantor, Luks, and Mark \cite{KantorLuksMark} exhibited an $\textsf{NC}$ algorithm for computing such a transversal, provided $|G| \leq n^{O(\log^c n)}$ for some fixed constant $c \geq 1$. They claimed to have a solution for the general problem (no restrictions on $|G|$), indicating that such a solution would appear in an in-preparation paper they cited as [KL2]. We were unable to locate such a solution, and to the best of our knowledge, this solution has not appeared.

\begin{remark} \label{rmk:Quotients}
It remains an intriguing open question whether our polynomial-time bounds for quotients (Theorem~\ref{thm:MainFittingFree}(a)) can be improved to $\textsf{NC}$. In the setting of polynomial-time computation, Kantor and Luks \cite{KantorLuksQuotients} conjectured a Quotient Group Thesis, stating:
\begin{quote}
``If a problem for quotient groups $G/K$ of permutation groups has a polynomial-time solution when $K = 1$, then it has a polynomial-time solution in general."    
\end{quote}

In the same paper, Kantor and Luks gave considerable supporting evidence for their Quotient Group Thesis by exhibiting a number of polynomial-time algorithms for quotient groups. There has been considerable work on quotient groups since, strengthening the case for the Quotient Group Thesis \cite{WilsonDirectProductsArxiv, BMWGenus2, LW12, LuksMiyazaki, DietrichWilsonCubeFree}. In some cases, such as computing direct product decompositions, access to quotients appears necessary \cite{WilsonDirectProductsArxiv}.

It is worth noting that while the polynomial-time library for quotients is similar to that of permutation groups, the techniques are often quite different. In the setting of permutation groups, algorithmic techniques often crucially leverage the action of the group on the permutation domain.  For quotient groups, such an action is not available \cite{neumann1985}, and so the techniques rely more deeply on the theory of the underlying group.

The setting of $\textsf{NC}$ highlights this distinction. We note that in order to obtain $\textsf{NC}$ bounds, we are limited to polylogarithmic iterations. However, at each iteration, we can perform a ``polynomial amount of work" in parallel (see Section~\ref{sec:Complexity} for a precise formulation of $\textsf{NC}$). Consider the problem of computing the centralizer $C_{G}(H)$, provided $G$ normalizes $H$. Babai, Luks, and Seress \cite{BabaiLuksSeress} exhibited an $\textsf{NC}$ algorithm for this problem in the setting of permutation groups, crucially taking advantage of the action of $G$ and $H$ on the underlying permutation domain. In the setting of quotient groups, Kantor and Luks \cite{KantorLuksMark} exhibited a polynomial-time algorithm for this problem by utilizing polynomial-time subroutines to compute Sylow subgroups \cite{KantorLuksMark,MarkThesis}, as well as the intersection $G \cap P$ between an arbitrary permutation group $G$ and a $p$-group $P$ \cite{LUKS198242,  LuksMiyazaki}. 

While Sylow subgroups of permutation groups can be computed in $\textsf{NC}$ \cite{KantorLuksMark, MarkThesis}, obtaining such bounds for the intersection problem remains elusive. Such procedures for computing $G \cap P$ employ a divide-and-conquer strategy, processing the orbits sequentially and then breaking the orbits up into blocks of imprimitivity. Note that there exist permutation groups of degree $n$ with a linear number of orbits (take, for instance, the action of $\mathbb{Z}_{2}^{n}$ on $2n$ points, which yields $n$ orbits of size $2$). Thus, the key bottleneck lies in the sequential processing of the orbits. Overcoming this obstacle appears to require new techniques, and so improving the bounds on Theorem~\ref{thm:MainFittingFree}(a) from polynomial-time to $\textsf{NC}$ is non-trivial. It is an intriguing open question as to whether the Quotient Group Thesis of Kantor and Luks holds in the setting of $\textsf{NC}$. 
\end{remark}

\begin{remark}
We also compare Theorem~\ref{thm:MainFittingFree} to the work of Cannon and Holt \cite{CH03}, for practical isomorphism testing of permutation groups. Cannon and Holt start with a permutation group $G$, and compute the solvable radical $\rad(G)$. They then deal with the Fitting-free quotient $G/\rad(G)$ by explicitly constructing a faithful permutation representation. Given a quotient group $\textbf{G} := G/K$, with $K \trianglelefteq G \leq \Sym(n)$, we can compute $\rad(\textbf{G})$ in polynomial-time \cite{KantorLuksQuotients}. Combining this with Theorem~\ref{thm:MainFittingFree}(a) formalizes that much of the framework established by Cannon and Holt \cite{CH03} is in fact polynomial-time computable. On the other hand, Grochow, Johnson, and Levet \cite{GrochowJohnsonLevet} exhibited a reduction from \algprobm{Linear Code Equivalence} (and hence, \algprobm{Graph Isomorphism}) to isomorphism testing of Fitting-free groups given by generating sequences of permutations. While \algprobm{Graph Isomorphism} is known to be solvable in quasipolynomial-time \cite{BabaiGraphIso}, the best known bound for \algprobm{Linear Code Equivalence} is $(2+o(1))^n$ \cite{BCGQ}. This provides formal evidence that the framework of Cannon and Holt \cite{CH03} is unlikely to run in polynomial time in the worst case. 

It remains open whether the solvable radical can be computed in $\textsf{NC}$, even in the setting of permutation groups.
\end{remark}

\noindent \textbf{Further Related Work.} Fitting-free groups have received considerable attention from the computational complexity community. Babai and Beals brought this class of groups to the attention of the computational complexity community in their work on black-box groups \cite{BabaiBeals}. In the multiplication (Cayley) table model, isomorphism testing of Fitting-free groups is known to be in $\textsf{P}$, through a series of two papers \cite{BCGQ, BCQ}, and this bound was recently improved to $\textsf{AC}^{3}$ \cite{GrochowJohnsonLevet}. This class of groups, as well as important sub-classes such as direct products of non-abelian simple groups and almost simple groups, have recently received attention from the perspective of the Weisfeiler--Leman algorithm \cite{BrachterSchweitzerWLLibrary, GLDescriptiveComplexity, GLWL1, CollinsLevetWL, BrachterThesis, GrochowJohnsonLevet, JLVWQuasigroups}. 


\section{Preliminaries}\label{sec:prelims}

\subsection{Groups}
All groups are assumed to be finite. Let $\text{Sym}(n)$ denote the symmetric group on $n$ letters. The \textit{normal closure} of a subset $S \subseteq G$, denoted $\text{ncl}_{G}(S)$, is the smallest normal subgroup of $G$ that contains $S$. The \textit{socle} of a group $G$, denoted $\text{Soc}(G)$, is the subgroup generated by the minimal normal subgroups of $G$. If $G$ has no abelian normal subgroups, then $\text{Soc}(G)$ decomposes uniquely as the direct product of non-abelian simple factors. The \textit{solvable radical} $\text{Rad}(G)$ is the unique maximal solvable normal subgroup of $G$. Every finite group can be written as an extension of $\text{Rad}(G)$ by the quotient $G/\text{Rad}(G)$, the latter of which has no abelian normal subgroups. We refer to the class of groups without abelian normal subgroups as \textit{Fitting-free}.

Let $G_1$ and $G_2$ be two groups acting on $\Delta$ and $\Gamma$ respectively. Let $\Gamma=\{\gamma_1,\ldots,\gamma_m\}$. Let $B$ be the set of functions from $\Gamma$ to $G_1$ (i.e., $\mathrm{Fun}(\Gamma,G_1)$). It is easy to see that $B$ forms a group with binary operation composition on two functions (see e.g., \cite[Section 2.6]{Dixon1996}). Then $B$ is isomorphic to the direct product of  $|\Gamma|$ copies of $K$ via isomorphism $f \mapsto (f(\gamma_1),\ldots,f(\gamma_m))$. Let $\Omega$ be the set of all functions from $\Gamma$ to $\Delta$. Then the \emph{wreath product} of $G_1$ by $G_2$ with respect to the set $\Gamma$, denoted by $G_1 \wr G_2 $, is defined to be the semidirect product of $B \rtimes G_2$ via $f^{g_2}(\gamma):=f(\gamma^{g_2^{-1}})$ for all $f \in \mathrm{Fun}(\Gamma,G_1), \gamma \in \Gamma$ and $g_2 \in G_2$. Clearly, $|G_1 \wr G_2|=|G_1|^{m}|G_2|$.

\subsection{Computational Complexity} \label{sec:Complexity}
We assume that the reader is familiar with standard complexity classes such as $\textsf{P}, \textsf{NP}, \textsf{L}$, and $\textsf{NL}$. Let $\textsf{FL}$ denote the class of logspace-computable functions. For a standard reference on circuit complexity, see \cite{VollmerText}. We consider Boolean circuits using the \textsf{AND}, \textsf{OR}, \textsf{NOT}, and \textsf{Majority}, where $\textsf{Majority}(x_{1}, \ldots, x_{n}) = 1$ if and only if $\geq n/2$ of the inputs are $1$. Otherwise, $\textsf{Majority}(x_{1}, \ldots, x_{n}) = 0$. In this paper, we will consider $\textsf{DLOGTIME}$-uniform circuit families $(C_{n})_{n \in \mathbb{N}}$. For this,
one encodes the gates of each circuit $C_n$ by bit strings of length $O(\log n)$. Then the circuit family $(C_n)_{n \geq 0}$
is called \emph{\textsf{DLOGTIME}-uniform}  if (i) there exists a deterministic Turing machine that computes for a given gate $u \in \{0,1\}^*$
of $C_n$ ($|u| \in O(\log n)$) in time $O(\log n)$ the type of gate $u$, where the types are $x_1, \ldots, x_n$, \textsf{NOT}, \textsf{AND}, \textsf{OR}, or \textsf{Majority} gates,
and (ii) there exists a deterministic Turing machine that decides for two given gates $u,v \in \{0,1\}^*$
of $C_n$ ($|u|, |v| \in O(\log n)$) and a binary encoded integer $i$ with $O(\log n)$ many bits
in time $O(\log n)$ whether $u$ is the $i$-th input gate for $v$.

\begin{definition}
Fix $k \geq 0$. We say that a language $L$ belongs to (uniform) $\textsf{NC}^{k}$ if there exist a (uniform) family of circuits $(C_{n})_{n \in \mathbb{N}}$ over the $\textsf{AND}, \textsf{OR}, \textsf{NOT}$ gates such that the following hold:
\begin{itemize}
\item The $\textsf{AND}$ and $\textsf{OR}$ gates take exactly $2$ inputs. That is, they have fan-in $2$.
\item $C_{n}$ has depth $O(\log^{k} n)$ and uses (has size) $n^{O(1)}$ gates. Here, the implicit constants in the circuit depth and size depend only on $L$.

\item $x \in L$ if and only if $C_{|x|}(x) = 1$. 
\end{itemize}
\end{definition}

Now the complexity class $\textsf{NC}$ is defined as:
\[
\textsf{NC} := \bigcup_{k \in \mathbb{N}} \textsf{NC}^{k}.
\]

\noindent We also allow circuits to compute functions by using multiple output gates. 

\begin{definition}
Fix $k \geq 0$. We say that a language $L$ belongs to (uniform) $\textsf{RNC}^{k}$ if there exists a polynomial $p_{L}(n)$ depending only on $L$, and a language $L'$ in (uniform) $\textsf{NC}^{k}$ such that the following hold. Let $(C_n)_{n \in \mathbb{N}}$ be a (uniform) family of circuits witnessing $L'$ belonging to $\textsf{NC}^k$.
\begin{itemize}
\item If $x \in L$, then:
\[
\text{Pr}_{y \in \{0,1\}^{\leq p(|x|)}} [C_{n}((x,y)) = 1] \geq 1/2, \text{ and }
\]

\item If $x \not \in L$, then for all $y \in \{0,1\}^{\leq p_{L}(|x|)}$, $C_{n}((x,y)) = 0$.
\end{itemize}

\noindent \\ Now define:
\[
\textsf{RNC} := \bigcup_{k \in \mathbb{N}} \textsf{RNC}^{k}.
\]
\end{definition}


\subsection{Permutation Group Algorithms}

In this paper, we will consider groups specified succinctly by generating sequences. The first model we will consider is that of permutation groups. Here, a group $G$ is specified by a sequence $S$ of permutations from $\Sym(n)$. 

Alternatively, we will consider the more general quotients of permutation groups model, which we will refer to as the \textit{quotients model}. To the best of our knowledge, this model was first introduced by Kantor and Luks \cite{KantorLuksQuotients}. Here, we will be given $K \trianglelefteq G \leq \Sym(n)$, represented in the following way:
\begin{itemize}
    \item $G, K$ are both specified by a generating sequence of permutations in $\Sym(n)$, and 
    \item Each element of $\textbf{G} := G/K$ is specified by a single coset representative.
\end{itemize}

The computational complexity for both the permutation group and the quotients models will be measured in terms of $n$. We recall a standard suite of problems with known $\textsf{NC}$ solutions in the setting of permutation groups.

\begin{lemma}[{\cite{BabaiLuksSeress}}]\label{PermutationGroupsNC}\label{lem:Order-NC}
Let $G$ be a permutation group. The following problems are in $\textsf{NC}$:
\begin{enumerate}[label=(\alph*)]
\item Compute the order of $G$.
\item Whether a given permutation $\sigma$ is in $G$; and if so, exhibit a word $\omega$ such that $\sigma = \omega(S)$. 

\item Find the kernel of any action of $G$.
\item Find the pointwise stabilizer of $B \subseteq [m]$.
\item Find the normal closure of any subset of $G$.

\item For $H \leq \Sym(n)$ such that $G$ normalizes $H$, compute $C_{G}(H)$.

\item For $H \leq \Sym(n)$ such that $G$ normalizes $H$, compute $G \cap H$.
\end{enumerate}
\end{lemma}

\noindent We easily extend a subset of these problems to the setting of quotients. 

\begin{lemma} \label{QuotientsNC}
Let $\textbf{G} = G/K$ be a quotient of a permutation group. The following problems are in $\textsf{NC}$:
\begin{enumerate}[label=(\alph*)]
\item Compute the order of $\textbf{G}$.

\item Given a permutation $\sigma$, test whether $\sigma K$ is in $\textbf{G}$; and if so, exhibit a word $\omega$ that evaluates to $\sigma K$. 
\end{enumerate}
\end{lemma}

\begin{proof}
We proceed as follows.
\begin{enumerate}[label=(\alph*)]
\item As $G, K$ are permutation groups, we can compute $|G|, |K|$ in $\textsf{NC}$ \cite{BabaiLuksSeress}.

\item We note that $\sigma K \in \textbf{G}$ if and only if $\sigma \in G$. So we use the $\textsf{NC}$ constructive membership test from \cite{BabaiLuksSeress} to test whether $\sigma \in G$.  \qedhere
 \end{enumerate}
\end{proof}

\begin{lemma} \label{lem:List}
Let $c > 0$. Let $K \trianglelefteq G \leq \Sym(n)$, and let $\textbf{G} := G/K$. If $|\textbf{G}| \leq n^{c}$, then we can list $\textbf{G}$ in $\textsf{NC}$.
\end{lemma}

\begin{proof}
We begin by computing $|\textbf{G}|$ in $\textsf{NC}$ (Lemma~\ref{QuotientsNC}). Now let $S$ be the set of generators we are given for $\textbf{G}$. Wolf \cite{Wolf} showed that we can write down $\langle S \rangle$ in $O(\log |\textbf{G}|)$ rounds. We recall Wolf's procedure here. Let $T_{0} = S \cup \{1\}$. Now for $i \geq 0$, we construct $T_{i+1}$ by computing $xy$, for each $x, y \in T_{i}$. Our procedure terminates when $|T_{i}| = |\textbf{G}|$. 

We note that Wolf is considering quasigroups given by their multiplication tables, in which case we can multiply two elements in $\textsf{AC}^{0}$. However, we are dealing with quotients of permutation groups. In particular, any element of $\textbf{G}$ is represented as a permutation in $\Sym(n)$. Note that multiplying two permutations is $\textsf{FL}$-computable \cite{COOK1987385}. Thus, given $T_{i}$, we can compute $T_{i+1}$ in $\textsf{FL}$. The result follows.
\end{proof}

\subsection{Background on Simple Groups} \label{subsec: simplegroup}
In this section, we examine classical simple groups, exceptional simple groups of Lie type, and other essential concepts related to simple groups. For a comprehensive treatment of the relevant details concerning simple groups, the reader is referred to \cite[Section 2.2]{DharaPhDThesis}. For a more in-depth exploration, further references include \cite{carter-book}, \cite{KleidmanLiebeck}, and \cite{WLS}.

\noindent \subparagraph*{Classical Simple Groups.} Let $V$ be a vector space of dimension $d$ over $\mathbb{F}_{q}$, where $q=p^{e}$. We write $\mathrm{GL}(d,q)$ to denote the \emph{general linear group}, which is the group of all invertible linear transformations from $V$ to $V$. The \emph{special linear group} (denoted by $\mathrm{SL}(d,q)$) is a subgroup of $\mathrm{GL}(d,q)$ containing all invertible linear transformations with determinant $1$. The center $Z({\rm{GL}}(d,q))$ consists of all transformation of the form $T(x)=\beta x$ for $\beta (\neq 0) \in \mathbb{F}_{q}$. The factor group $\mathrm{SL}(d,q)/Z(\mathrm{SL}(d,q)$ is called the \emph{projective special linear group} and is denoted by $\PSL(d,q)$. It is known that $\PSL(d,q)$ is a simple group with few exceptions \cite{carter-book,WLS}.

Now we define the simple group $\PSp(2d, q)$. Let $V$ be a vector space of dimension $2d$ over $\mathbb{F}_{q}$, where $q=p^{e}$. The symplectic classical group $\mathrm{Sp}(2d,q)$ is the set of invertible linear transformations of a $2d$-dimensional vector space over the field $\mathbb{F}_{q}$ which preserves a non-degenerate, alternating, skew-symmetric bilinear form (see, e.g., \cite[Page~13]{KleidmanLiebeck}). It is known that any non-degenerate, skew-symmetric, bilinear form can be represented with respect to some basis by the matrix \[X= \begin{bmatrix}
    0 & 1 & & & & & & & \\
    -1 & 0 & &  & &   & \bf{0} & &  \\
   & & 0 & 1  & &  & & &\\
 & & -1 & 0 & &  & & & \\
   & \bf{0}  & & & &\ddots & & & \\
   & & & & & & & 0 & 1   \\
   & & & & & & & -1 & 0   \\
\end{bmatrix}.\] 

Let $T$ be the matrix corresponding to a linear transformation $\tau$. We have that $\tau \in {\rm{Sp}}(2d,q)$ if and only if $T^tXT=X$, where $T^t$ is the transpose of $T$. The center  $Z({\rm{Sp}}(2d,q))$ consists of the linear transformations $f(x)=\beta x$ where $\beta=\pm 1$. The factor group ${\mathrm{Sp}}(2d,q)/Z({\rm{Sp}}(2d,q))$ is called \emph{projective symplectic group} $\PSp(2d,q)$. The group $\PSp(2d,q)$ is a simple  provided ${d} \geq 2$ except when ${\rm PSp}(4,2)$. In the case when $\mathbb{F}_{q}$ is of characteristic $2$, we have $\PSp(2d,q)={\mathrm{Sp}}(2d,q)$.

Let us now define the simple group $\POmegaPlus(2d, q)$.  Let $V$ be vector space of dimension $m \geq 2$ over  $\mathbb{F}_{q}$, where $q=p^{e}$ and $p$ odd. Let $Q$ be a quadratic form whose associated bilinear form is $\mathfrak{f}_{Q}$. Then, the isometries of $V$ are called orthogonal linear transformations; they comprise the orthogonal group ${\rm{O}}(V, Q)$ \cite{LarryBook, WLS}. Thus 

\[{\rm{O}}(V,Q)=\{\tau \in \mathrm{GL}(m,q) | Q(\tau u)= Q(u) \text{ for all } u \in V\} \leq   \mathrm{GL}(m,q)\]

We denote by $\Omega(V,Q)$ the commutator subgroup of ${\rm{O}}(V,Q)$. Let $Z({\rm{O}}(V,Q))$ be the center of ${\rm{O}}(V,Q)$. Then the corresponding projective group is defined as \[\mathrm{P\Omega}(V,Q)=\Omega(V,Q)/(Z({\rm{O}}(V,Q)) \cap \Omega(V,Q)).\]

Now assume that $m=2d$. In this case, there are two non-equivalent quadratic forms on $V$, say $Q$ and $Q'$; these forms give rise to two different orthogonal groups (see, e.g., section 1.4, \cite{carter-book}, \cite{KleidmanLiebeck}). These quadratic forms are represented by the following matrices $X$ and $X'$, respectively,

$$X=\begin{bmatrix}
    0 & I_{d} \\
    I_{d} & 0
\end{bmatrix} \text{ and } X'= \begin{bmatrix}
    0_{d-1} & & & I_{d-1} & 0 & 0\\
    
    &  & & & \vdots & \vdots \\
    I_{d-1} &  &  & 0_{d-1}  & 0 & 0\\
    0 & & \ldots & 0 & 1 & 0 \\
    0 & & \ldots & 0 & 0 & - \epsilon 
\end{bmatrix},$$

where $\epsilon$ is a non-square element in $\mathbb{F}_{q}$. The orthogonal groups corresponding to these quadratic forms are denoted by
$O^{+}(2d,q)=O(V,Q)$ and $O^{-}(2d,q)=O(V,Q')$, respectively. Similarly, $\Omega^{+}(2d, q)$ is the commutator subgroup of $O^{+}(2d, q)$, and $\Omega^{-}(2d, q)$ is the commutator subgroup of $O^{-}(2d, q)$. The associated simple groups corresponding to $O^{+}(2d,q)$ and $O^{-}(2d,q)$ are denoted by $\mathrm{P\Omega}^{+}(2d,q)$ and  $\mathrm{P\Omega}^{-}(2d,q)$, respectively (see, e.g., \cite[Pages~4-7]{carter-book}). The groups ${\rm P}\Omega^{+}(2{d},q)$ and ${\rm P}\Omega^{-}(2{d},q)$ are simple when ${d} \geq 4$.

A subspace $W$ of $V$ is said to be \emph{anisotropic} with respect to a quadratic form $Q$ if the restriction of $Q$ to the subspace is nondegenerate. A bilinear map $\mathfrak{f}_Q$ is the null-form if $\mathfrak{f}_Q(u,v)=0$ for all $u,v \in V$. A subspace $W$ of $V$ is said to be \emph{isotropic} with respect to quadratic form $Q$ if the restriction of $\mathfrak{f}_Q$ to $W$ is the null-form and $Q(w)=0$ for all $w \in W$ \cite{ref4MV, ref7VClassical}.

Now we describe various automorphisms of $\SL(d,q)$, $\Omega^+(2d,3)$ ($d \geq 4$), or $\mathrm{Sp}(4, 2^{e})$ ($e \geq 2$)(see \cite{carter-book} for more details).

\begin{definition}\label{def:automorphisms}
Let $G$ be one of the following: $\SL(d,q)$, $\Omega^+(2d,3)$ ($d \geq 4$), or $\mathrm{Sp}(4, 2^{e})$ ($e \geq 2$). We introduce several automorphisms.
\begin{enumerate}[label=(\alph*)]
\item An automorphism $\mathbbm{i} : G \to G$ is called an \emph{inner automorphism} if there exists a matrix $F \in G$ such that $\mathbbm{i}(U) = F^{-1}UF$ for all $U \in G$.

\item An automorphism $\mathbbm{d} : G \to G$ is called a \emph{diagonal automorphism} if there exists a diagonal matrix $F \in \text{GL}(d,q)$ such that $\mathbbm{d}(F) = F^{-1}UF$ for all $U \in G$.

\item Let $\sigma$ be an automorphism of the field $\mathbb{F}_{q}$. A \emph{field automorphism} $\mathbbm{f}_{\sigma} : G \to G$ is defined as $\mathbbm{f}_{\sigma}([u_{ij}]) = [\sigma(u_{ij})]$, where $[u_{ij}]$ denotes the matrix in $G$ whose $(i,j)$th entry is $u_{ij}$.

Let $q=p^e$. Let us consider the Frobenius automorphism $\sigma_0$ of $\aut(\mathbb{F}_q)$ given by $a \mapsto a^p$. The automorphism $\sigma_0$ is a generator of $\aut(\mathbb{F}_q)$ of order $e$. Therefore, $\mathbbm{f}_{\sigma}=\mathbbm{f}_{\sigma_0}^{t'}$ for some $1 \leq t' \leq e$.

\item The \emph{graph automorphism} $\mathbbm{g} : G \to G$ is defined as $\mathbbm{g}(U) = (U^{t})^{-1}$ for all $U \in G$. Here, $U^{t}$ denotes the transpose of $U$.
\end{enumerate}
\end{definition}

A map $f: V \rightarrow V$ is called \emph{semi-linear transformation} of $V$ if there is an automorphism $\sigma \in \aut(\mathbb{F}_{q})$ such that for all $u,v \in V$ and $\beta \in \mathbb{F}_{q}$, $f(u+v)=f(u)+f(v)$ and $f(\beta v)=\sigma(\beta) f(v)$. A semi-linear transformation $f$ is invertible if $\{v\in V \mid f(v)=0\}=\{0\}$. The group of all invertible semi-linear transformations is called \emph{general semi-linear group}, denoted by $\Gamma{\rm L}(d,q)$ \cite{KleidmanLiebeck}. The \emph{projective general semi-linear group}, ${\rm P}\Gamma{\rm L}(d,q)$, is obtained by factoring out the scalars in $\Gamma{\rm L}(d,q)$

However, for the current purpose the only fact that we need regarding ${\rm P}\Gamma{\rm L}(d,q)$ is that it can be regarded as a subgroup of $\aut(\SL(d,q))$ and automorphisms of $\SL(d,q)$ that contain the graph automorphisms $\mathbbm{g}$ in its decomposition are exactly the elements in $\aut(\SL(d,q)) \setminus {\rm P}\Gamma{\rm L}(d,q)$ \cite{KleidmanLiebeck,carter-book}.

\noindent \subparagraph*{Exceptional Groups of Lie Type.} In this section, we define the Chevalley generators for $\mathcal{L}(\mathbb{F}_{q}) \in \{\G(3^e), \F(2^e), \Esix(q)\}$. Unlike the classical simple groups, the exceptional simple groups of Lie type do not have matrix representations. Thus, we opt for this Chevalley presentation. Additionally, the Chevalley generator framework provides a well-established method for studying the automorphisms of these groups. Let us now define the Chevalley generators of the simple group $\mathcal{L}(\mathbb{F}_{q})$ from \cite[Chapter 3-4, 7]{carter-book}.

Let $\mathcal{L}$ be a Lie algebra over $\mathbb{C}$. Let $\Phi$ be a root system corresponding to $\mathcal{L}$ and $\Pi \subsetneq \Phi$ be a fundamental system of roots. Let $\mathcal{L}_{\mathbb{Z}}$ be the subset of $\mathcal{L}$ of all linear combinations of basis elements with coefficients in the ring $\mathbb{Z}$ of integers. Define $\mathcal{L}_{\mathbb{F}_q}$ using the tensor product as follows: $\mathcal{L}_{\mathbb{F}_q}= \mathbb{F}_q  \otimes \mathcal{L}_{\mathbb{Z}}$ (see e.g., \cite[Section~4.4, Page 62]{carter-book}). The automorphisms of $\mathcal{L}_{\mathbb{F}_q}$ are the Chevalley generators of $\mathcal{L}(\mathbb{F}_{q})$. For each $r \in \Phi$ and $\beta \in \mathbb{F}_q$, $x_{r}(\beta)$ is an automorphisms of $\mathcal{L}_{\mathbb{F}_q}$. The action of $x_{r}(\beta)$ on the generating set of $\mathcal{L}_{\mathbb{F}_q}$ has been explicitly defined in \cite[Chapters~3-4, 7]{carter-book} for the comprehensive details.

We have $\mathcal{L}(\mathbb{F}_{q}) = \langle \{x_{r}(\beta) \mid r \in \Phi \text{ and } \beta \in \mathbb{F}_{q}\} \rangle$ (see e.g., \cite[Section~4.4, Page~64]{carter-book}). E.g., let $\mathcal{L}=\G$, then $\mathcal{L}(\mathbb{F}_{q})= \G(q)$ is the exceptional group of Lie type obtained from $\mathcal{L}$. Also, we have $\G(q) = \langle \{ x_{r}(\beta) \mid r \in \Phi \text{ and } \beta \in \mathbb{F}_q\} \rangle$ and $|\Phi|=12$. Similarly, we can define the Chevalley generators for $\F(2^e), \Esix(q)$. In case of $\F(2^e)$, $|\Phi|=2^7 3^2$, and in case of $\Esix(q)$, we have $|\Phi|=2^7 3^4 5$ (see  \cite[Page~43]{carter-book}). Therefore, the size of the Chevalley generators is polynomial in $q$. Additionally, note that, since $|\Phi|$ is bounded by a constant, the size of a generating set $\{h_{r}, r \in \Pi, e_r, r \in \Phi \}$ of $\mathcal{L}_{\mathbb{F}_q}$ is also bounded by a constant.

Let us now describe the automorphisms of the Chevalley group $\mathcal{L}(\mathbb{F}_{q})$ (see e.g., \cite[Chapter~12]{carter-book}). Notice that it is enough to define the automorphism at the Chevalley generators of the group $\mathcal{L}(\mathbb{F}_{q})$. 

There are four types of automorphism of a finite simple group: the inner automorphism, the diagonal automorphism, the field automorphism, and the graph automorphism. Any automorphism of $\mathcal{L}(\mathbb{F})$ can be decomposed as a product of these four automorphisms \cite{carter-book}. We first define the inner and diagonal automorphisms of $\mathcal{L}(\mathbb{F}_{q})$.

Let $P=\mathbb{Z}\Phi$ be the set of all linear combinations of elements of $\Phi$ with rational integer coefficients. $P$ is the additive group generated by the roots of $\mathcal{L}$. A homomorphism from $P$ into the multiplicative group $\mathbb{F}^{*}_{q}$ is called $\mathbb{F}_{q}$-character of $P$. The $\mathbb{F}_{q}$-character of $P$ forms a multiplicative group and each $\mathbb{F}_{q}$-character $\chi$ gives rise to an automorphism $h(\chi)$ of $\mathcal{L}_{\mathbb{F}_{q}}$. The automorphisms $\mathcal{L}_{\mathbb{F}_{q}}$ of the form $h(\chi)$ form a subgroup $\hat{H}$ of the full automorphism group of $\mathcal{L}_{\mathbb{F}_{q}}$ (see e.g., \cite[Section~7.1, Page~98]{carter-book}). Let $H$ be the subgroup of $\hat{H}$ generated by the automorphisms $h_{r}(\beta)$ for all $r\in \Phi$, $\beta \in \mathbb{F}^{*}_{q}$. Then $G=\mathcal{L}(\mathbb{F}_{q})$ normalized by $\hat{H}$ in the group of all automorphisms of $\mathcal{L}_{\mathbb{F}_{q}}$. Thus, if $h(\chi) \in \hat{H}$ and $g \in G$, then a map $g \rightarrow h(\chi) g h(\chi)^{-1}$ is an automorphism of $G$. If $h(\chi) \in \hat{H} \setminus H$, then the automorphism is called a \emph{diagonal automorphism} of $G$. If $h(\chi) \in H$, then the automorphism is called an \emph{inner automorphism} of $G$ (see e.g., \cite[Section~12.2, Page~200]{carter-book}).

Any automorphism $\sigma_0$ of the field $\mathbb{F}_{q}$ induces a field automorphism, say $\mathbbm{f}_{\sigma_0}$, of $\mathcal{L}(\mathbb{F}_{q})$.
It is defined in the following manner: $\mathbbm{f}_{\sigma_0}: \mathcal{L}(\mathbb{F}_{q}) \rightarrow \mathcal{L}(\mathbb{F}_{q})$, $\mathbbm{f}_{\sigma_0}(x_{r}(\beta))=x_{r}(\sigma_0(\beta))$ for $r \in \Phi$. The set of field automorphisms of $\mathcal{L}(\mathbb{F}_{q})$ is a cyclic group isomorphic to $\aut(\mathbb{F}_{q})$ (see e.g., \cite[Section~12.2, Page~200]{carter-book}).

The graph automorphisms of the Chevalley group $\mathcal{L}(\mathbb{F}_{q})$ arise from the symmetries of the Dynkin diagram. A symmetry of the Dynkin diagram of $\mathcal{L}$ is a permutation $\varrho$ of the nodes of the diagram such that the number of bonds joining nodes $i.j$ is the same as the number of bonds joining $\varrho(i),\varrho(j)$ for any $i\neq j$. Let $r \in \Phi$, $\beta \in \mathbb{F}_{q}$, and let $r \rightarrow {\bar r}$ be a map of $\Phi$ into itself arising from a symmetry of the Dynkin diagram of $\mathcal{L}$ (see e.g., \cite[Sections~12.2-12.4]{carter-book} and \cite[Section~2]{auto-Chevalley-groups}).  The graph automorphisms $\mathbbm{g}$ of the Chevalley groups $\mathcal{L}(\mathbb{F}_{q}) \in \{\G(3^e), \F(2^e), \Esix(q)\}$ are defined in the following ways (\cite{carter-book}): 

\begin{enumerate}
    \item if $\mathcal{L}(\mathbb{F}_{q})=\G(3^e)$,
    $$  \mathbbm{g} (x_{r}(\beta))=x_{\bar {r}}(\beta^{\lambda(\bar {r})}),$$
    where $\lambda(\bar {r_i})=1$ if $r$ is short root and $3$ if $r$ is long root.
    
    \item $\mathcal{L}(\mathbb{F}_{q})=\F(2^e)$,
    $$  \mathbbm{g} (x_{r}(\beta))=x_{\bar {r}}(\beta^{\lambda(\bar {r})}),$$
    where $\lambda(\bar {r_i})=1$ if $r$ is short root and $2$ if $r$ is long root.
    
    \item $\mathcal{L}(\mathbb{F}_{q})=\Esix(q)$,
    $$  \mathbbm{g} (x_{r}(\beta))=x_{\bar {r}}(\gamma_{r} z),$$
    where $\gamma_{r} \in \mathbb{Z}$ can be chosen so that $\gamma_{r}=1$ if $r \in \Pi$, $\gamma_{r}=-1$ if $-r \in \Pi$.
\end{enumerate}

\begin{theorem}(cf. \cite[Theorem 12.5.1, p. 211]{carter-book})\label{carter Thm-decomposition}
Let $G$ be a classical simple group or an exceptional simple group of Lie type defined over a field $\mathbb{F}$. Then every automorphisms of $G$ can be decomposed as $ \mathbbm{idgf}$, where $\mathbbm{i}$, $\mathbbm{d}$,  and $\mathbbm{f}$ are inner, diagonal, field automorphisms respectively, and $\mathbbm{g}$ is the graph automorphism.
\end{theorem}

\begin{remark}\label{baarnhielm Thm-decomposition}
Note that every automorphism of $G$ can also be decomposed as $\mathbbm{idfg}$ as described in \cite[Section 10]{baarnhielm2015}.
\end{remark}

\subsection{The Minimal Faithful Permutation degree of Simple and Almost simple groups} \label{sec:ConstructiveRecognition}

In this section, we will examine how to construct minimal faithful permutation representations of finite simple groups.  We will crucially use the following result of Kantor.

\begin{theorem}[{\cite{kantorPrimeOrderElement}}] \label{thm:Kantor}
If $G \leq \Sym(\Omega)$, with $|G| \geq |\Omega|^9$, then $G$ and $\Omega$ are (up to a permutation isomorphism):
\begin{itemize}
\item $G = \text{Alt}(r)$ for some $r$, and there exists some $k$ such that $\Omega$ is the set of all $k$-element subsets of $[r]$;

\item $G = \text{Alt}(r)$ for some $r$, and $\Omega$ is the set of partitions of the underlying set $[r]$ into blocks of size $k$ for some $k$; or 

\item $G$ is a classical group, and $\Omega$ is an orbit of subspaces of an underlying vector space.
\end{itemize}
\end{theorem}

In our setting, Theorem~\ref{thm:Kantor} essentially allows us to focus on classical simple groups $G$, where the underlying vector space $V$ satisfies $\text{dim}(V) > 8$ (see e.g., \cite{kantorPrimeOrderElement, KantorLuksMark}). Thus, we are able to avoid technicalities when the dimension of the underlying vector space is \emph{small}. Additionally, if the  simple group in question has order at most $n^9$, then we may write down its multiplication table in $\textsf{NC}$, using Lemma~\ref{lem:List}. 

We now recall key facts concerning minimal faithful permutation representations of certain classical simple groups. Let $V = \mathbb{F}_{q}^{m}$, and let $\overline{V}$ be an (orbit of) $1$-spaces of $V$. The \emph{natural action} of a classical group $G$ on $\overline{V}$ takes $g \in G$ and $[v] \in \overline{V}$, and sends $[v] \mapsto [gv]$. 

\begin{theorem} \label{thm:MinPerRepClassical}
We have the following.
\begin{enumerate}[label=(\alph*)]
\item Let $G = \PSL(d,q)$, with $d \geq 3$ and $(d,q) \neq (3,2), (4,2)$. Then $\mu(G) = (q^d - 1)/(q-1)$. Let $\overline{V}$ be the set of all $1$-spaces of $\mathbb{F}_{q}^{d}$. The natural action of $G$ on $\overline{V}$ yields a minimal faithful permutation representation of $G$ (see e.g., \cite{ref7VClassical, KantorLuksMark, MarkThesis}).

\item Let $G = \PSp(4, 2^e)$ with $e \geq 2$. Then $\mu(G) = (2^{4e} - 1)/(2^{e}-1)$. Let $\overline{V}$ be the set of all $1$-spaces of $\mathbb{F}_{2^e}^{4}$. The natural action of $G$ on $\overline{V}$ yields a minimal faithful permutation representation of $G$ (see e.g., \cite{ref7VClassical, KantorLuksMark, MarkThesis}).

\item Let $G = \POmegaPlus(2d,q)$, where (i) $d = 4$ and $q \geq 4$; or (ii) $d \geq 4$ and $q = 3$. Let $Q$ be the corresponding quadratic form, and $\mathfrak{f}_{Q}$ be the associated bilinear form. Let $\overline{V}$ be the set of $1$-spaces of $\mathbb{F}_{q}^{2d}$ that are anisotropic with respect to $\mathfrak{f}_{Q}$, spanned by vectors with equal values on $Q$. The natural action of $G$ on $\overline{V}$ yields a minimal faithful permutation representation of $G$ (see e.g., \cite{ref4MV, KantorLuksMark, MarkThesis}).
\end{enumerate}
\end{theorem}

\noindent For the remaining classical simple groups, we will not need to know detailed descriptions of their minimal faithful permutation representations on the appropriate (orbit of) $1$-spaces of the underlying vector space $V$. We refer to \cite{ref7VClassical} for an overview of the linear, symplectic, and unitary groups; and \cite{ref4MV} for the orthogonal simple groups.

We first recall the following result due to Kantor, Luks, and Mark \cite{KantorLuksMark} (see \cite{MarkThesis} for more details).

\begin{theorem} \label{thm:ConstructiveRecognition}
Let $K \trianglelefteq G \leq \Sym(n)$, and let $\mathbf{G} := G/K$. If $\mathbf{G}$ is a non-abelian simple group, then we can, in $\textsf{NC}$, determine the name $S$ of $\mathbf{G}$, as well as construct an isomorphism $\varphi$ from $\mathbf{G}$ to an isomorphic copy $\mathbf{G}_{1}$ of $\mathbf{G}$. Furthermore, $\varphi$ is a minimal permutation representation of $\mathbf{G}$. 

Suppose first that $\textbf{G}$ is a sporadic group, exceptional group of Lie type, an alternating group of order less than $n^9$, or a classical simple group of dimension $m$ over $\mathbb{F}_{q}$ of order less than $n^9$ that does not admit a minimal faithful permutation representation over an appropriate orbit of $1$-spaces from $\mathbb{F}_{q}^{m}$. Then in these cases, $\varphi$ details the action of $\textbf{G}$ on the cosets $\textbf{G}/\textbf{P}$ of a maximal subgroup $\textbf{P}$. In particular, our procedure returns $\textbf{P}$ explicitly, including a $4$-element generating sequence.

Suppose instead that $\textbf{G}$ is a classical simple group, then  of dimension $m$ over $\mathbb{F}_{q}$ that admits a minimal faithful permutation representation detailing the action of $\textbf{G}$ on an appropriate orbit of $1$-spaces of $\mathbb{F}_{q}^{m}$. Then the following hold:

\begin{enumerate}
\item $\varphi$ is a minimal faithful permutation representation, detailing the action of $\textbf{G}$ on the appropriate orbit $\overline{V}$ of $1$-spaces of $\mathbb{F}_{q}^{m}$. In particular, if $\textbf{G}$ is one of the cases outlined in Theorem~\ref{thm:MinPerRepClassical}, then the permutation representation is that described in Theorem~\ref{thm:MinPerRepClassical}.\footnote{For $m \geq 4$ and $q \geq 3$, the action of $\POmegaPlus(2m, q)$ on the set $\Delta_{C}$ of anisotropic $1$-spaces of the underlying vector space $\mathbb{F}_{q}^{2m}$. When $q > 3$, $|\Delta_{C}| = (q^{m}-1)(q^{m-1}+1)/(q-1)$. When $q = 3$, $|\Delta_{C}| = 3^{m-1}(3^{m}-1)/2$ \cite{ref4MV}. Kantor, Luks, and Mark \cite{KantorLuksMark, MarkThesis} recognize $\POmegaPlus(2m,q)$ by constructing the action $\POmegaPlus(2m,q)$ on $\Delta_{C}$. In particular, they construct a minimal faithful permutation representation of $\POmegaPlus(2m, q)$ when $q \geq 3$.

However, in the procedure to determine the name of the finite simple group \cite[Procedure for Problem V.55]{MarkThesis}, the case $\POmegaPlus(2m, q)$ is incorrectly handled by checking whether $|\Delta_{C}| = (q^{m}-1)(q^{m-1}+1)/(q-1)$. This is easily fixed by checking whether $|\Delta_{C}| = 3^{m-1}(3^{m}-1)/2$ when $q = 3$ (in which case, we have $\POmegaPlus(2d,3$); or $|\Delta_{C}| = (q^{m}-1)(q^{m-1}+1)/(q-1)$ when $q > 3$ (in which case, we have $\POmegaPlus(2m, q)$ for $q > 3$).}

\item The elements of $\mathbf{G}_{1}$ are matrices acting on the vector space $V = \mathbb{F}_{q}^{m}$ of dimension $m \in O(\log n)$,
\item The isomorphism $\varphi$ comes with a basis $\beta$ for $V$ and a data structure that:
\begin{enumerate}
\item Given a permutation $g \in \textbf{G}$, we can, in $\textsf{NC}$, construct the corresponding matrix representation $T$ with respect to $\beta$, and 
\item Given a matrix $T \in \mathbf{G}_{1}$, we can, in $\textsf{NC}$, construct the corresponding permutation $g \in \textbf{G}$. 
\end{enumerate}
\end{enumerate}

\noindent Furthermore, if $\mathbf{G}$ is instead an alternating group of order at least $n^9$, then $\varphi$ details the appropriate corresponding action outlined in Theorem~\ref{thm:Kantor}.
\end{theorem}

We note that $\PSL(2,5) \cong \text{Alt}(5)$. While there are $6$ one-dimensional subspaces of $\mathbb{F}_{5}^{2}$, $\mu(\PSL(2,5)) = \mu(\text{Alt}(5)) = 5$. Thus, the hypothesis about a classical simple group admitting a minimum permutation representation over an appropriate orbit of $1$-spaces is indeed necessary. However, every sufficiently large classical simple group admits a minimal faithful permutation representation over an appropriate orbit of $1$-spaces of the underlying vector space \cite{ref4MV, ref7VClassical, kantorPrimeOrderElement}.

\begin{proof}[Proof of Theorem~\ref{thm:ConstructiveRecognition}]
We first show how to reduce to the case when $K = 1$ in $\textsf{NC}$. Kantor, Luks, and Mark \cite[Section~3.6]{KantorLuksMark} established that there exists an $\textsf{NC}$-computable series $K = N_0 \trianglelefteq N_1 \trianglelefteq \cdots \trianglelefteq N_{r} = G$, where for each $0 \leq i < r$, $N_{i+1}/N_{i}$ is either abelian or a direct product of non-abelian simple groups. Furthermore, their procedure constructs a faithful permutation representation of each non-abelian simple group. Thus, we may determine if $\textbf{G} = G/K$ is simple in $\textsf{NC}$; and if so, construct a faithful permutation representation of $\textbf{G}$. Thus, we may now assume without loss of generality that $K = 1.$ 

Now in $\textsf{NC}$, we may compute $|\textbf{G}|$ (Lemma~\ref{PermutationGroupsNC}(a)). If $|\textbf{G}| \geq n^9$, then the result follows from Kantor, Luks, and Mark \cite{KantorLuksMark, MarkThesis}. 

Suppose instead that $|\textbf{G}| < n^9$. If $\textbf{G}$ is a classical simple group of dimension $m$ over $\mathbb{F}_{q}$ that admits a minimal faithful permutation representation over an appropriate orbit of $1$-spaces of $\mathbb{F}_{q}^{m}$, then we use Lemma~\ref{lem:MinPerClassicalSimple}. Otherwise, we use Lemma~\ref{Minper-Simple}. The result now follows.
\end{proof}

\begin{lemma}\label{Minper-Simple}
Let $G, K \leq \Sym(n)$ with $K \trianglelefteq G$. Let $\textbf{S} = G/K$ be a simple group such that $|\textbf{S}| < n^9$ or $\textbf{G}$ is sporadic. Then in $\textsf{NC}$, we can compute a minimal faithful permutation representation $\varphi : \textbf{S} \to \Sym(\mu(\textbf{S}))$ detailing the action of $\textbf{S}$ on the cosets $\textbf{S}/\textbf{P}$ of a maximal subgroup $\textbf{P}$. In particular, our procedure returns a $4$-element generating sequence for $\textbf{P}$. 
\end{lemma}

\begin{proof}
As $|\mathbf{S}| < n^9$, we may write down the multiplication table for $\textbf{S}$ in $\textsf{NC}$ (Lemma~\ref{lem:List}). As $\mathbf{S}$ is simple, we have from \cite{JohnsonMPD} that there exists a maximal subgroup $\textbf{P} \leq \textbf{S}$ such that the regular action of $\textbf{S}$ on the right cosets of $\textbf{S}/\textbf{P}$ is a minimal faithful permutation representation. Any maximal subgroup of a finite simple group is $4$-generated \cite{BurnessLiebeckShalev}. As we have constructed the multiplication table for $\mathbf{S}$, we can use the generator-enumeration strategy to examine all $4$-generated subgroups of $\textbf{S}$ in $\textsf{FL}$ \cite{TangThesis}. We take the subgroup $\textbf{P}$ minimizing $[\mathbf{S} : \mathbf{P}]$. Note that as $\mathbf{S}$ is simple, $\text{core}_{\mathbf{S}}(\mathbf{P}) = 1$. The result now follows.
\end{proof}

\begin{lemma} \label{lem:MinPerClassicalSimple}
Let $G, K \leq \Sym(n)$ with $K \trianglelefteq G$. Let $\textbf{S} = G/K$ such that $|\textbf{S}| < n^9$. Suppose that $\textbf{S}$ is a classical group of dimension $d$ over $\mathbb{F}_{q}$, and that we are given its standard name (e.g., $\PSL(d,q)$). Furthermore, suppose that $\textbf{S}$ admits a minimal faithful permutation representation detailing the action of $\textbf{S}$ on an appropriate set of $1$-spaces of $V = \mathbb{F}_{q}^{d}$. Then we can in $\textsf{NC}$, do the following.
\begin{enumerate}[label=(\alph*)]
\item Construct $\mathbb{F}_{q}^{d}$ by explicitly listing its elements, as well as a basis $\beta$.

\item Compute a minimal faithful permutation representation $\varphi : \textbf{S} \to \Sym(\mu(S))$ detailing the action of $\textbf{S}$ on the appropriate orbit of $1$-spaces $\overline{V}$ of $V$.

\item Construct the underlying form.
\end{enumerate}
\end{lemma}

\begin{proof}
We may use a standard $2$-element generating sequence of $\textbf{S}$, where the generators are represented as $d \times d$ matrices over $V$. In particular, as $|\textbf{S}| < n^9$, we can write down the multiplication table in $\textsf{NC}$ using Lemma~\ref{lem:List} (note that as $q, d \leq n$, we can multiply two matrices in $\textsf{NC}$). In particular, we may identify each element in the multiplication table of $\textbf{S}$ with its corresponding matrix representation.

We may then, in $\textsf{NC}$, write down the vector space $V = \mathbb{F}_{q}^{d}$ by listing its elements. Additionally, as we are given the standard name of $\textbf{S}$, we may in $\textsf{NC}$ construct the underlying form. Once we have constructed the underlying form, then we may identify the appropriate orbit $\overline{V}$ of $1$-spaces of $V$ in $\textsf{NC}$ by evaluating the forms on the elements of $V$. We may then construct a minimal faithful permutation representation of $\textbf{S}$ by considering the action of each matrix on $\overline{V}$. The result now follows.    
\end{proof}

Let $S$ be a simple group, then a group $A$ with $S \leq A \leq \aut(S)$ is called an \emph{almost simple group}. A normal subgroup $H\neq 1$ of $G$ is called \emph{minimal normal} if the only normal subgroups of $G$ contained in $H$ are $1$ and $H$. A minimal normal subgroup $H$ of a finite group $G$ is either simple or a direct product of isomorphic simple groups (see e.g., \cite[Page~106]{rotman}). 

The following proposition, due to Cannon, Holt, and Unger, is one of the crucial ingredients for our results.

\begin{proposition}[\cite{CannonHoltUnger}]\label{MD-AlmostSimple}
Let $S \leq A \leq \rm{Aut}(S)$ for a finite non-abelian simple group $S$. Then $\mu(A)=\mu(S)$, except for the following cases described in the Table \ref{Table-MD-AlmostSimple}.

\begin{table}[H]

\caption{Table containing $\mu$ value of certain simple groups and almost simple groups}
\label{Table-MD-AlmostSimple}
  
{\begin{tabular}{ccccl}
    \toprule
    & $S$ & $A $ & $\mu(S)$ & $ \mu(A)$\\
    \midrule
    $1$ & $A_6 $ & $A \not \leq S_6 $ & $6$ & $10$ \\ 
  
$2$ & $\PSL(2,7)$ & $\mathrm{PGL}(2,7)$ & $7$ & $8$\\ 
 
$3$ & $\rm{M}_{12}$ & $\mathrm{Aut}(M_{12})=M_{12}.2$ & $12$ & $2 \mu(S)$ \\

$4$ & $\rm{O'N} $ & $\mathrm{Aut}(O'N)=O'N.2$ & $122760$ & $2 \mu(S)$ \\
 
$5$ &  $\mathrm{PSU}(3,5)$ & $A \not \leq \rm{P}\Sigma \rm{U}(3,5)$ & 50 & $126$\\ 
 $6$ & $\mathrm{P\Omega^{+}(8,2)}$ & $3 \mid | A /S|$ & $120$ & $3 \mu(S)$ \\  

$7$ & $\mathrm{P\Omega^{+}(8,3)}$ & $3 \mid | A /S|, 12\nmid | A /S|$ & $1080$ & $3 \mu(S)$\\

$8$ & $\mathrm{P\Omega^{+}(8,3)}$ & $12 \mid  |A/S|$ & $1080$ & $3360$\\ 

$9$ & $\G(3)$ & $\mathrm{Aut(G_2(3))}$ & $351$ & $2 \mu(S)$\\ 
$10$ & $\POmegaPlus(8,q)$, $q\geq 4 $ & $3 \mid | A /S|$&  $\frac{ (q^4-1)(q^3+1)}{q-1}$ & $3 \mu(S)$ \\

$11$& $\PSL(d,q)$, $d \geq 3$ $(d,q) \neq (3,2),(4,2)$ & $A \not \leq \PGammaL$ & $\frac{q^d -1 }{q-1}$ &$2 \mu(S)$  \\ 

$12$ & $\PSp(4,2^e)$, $e\geq 2 $ & $A \not \leq \mathrm{P\Gamma Sp}(4,2^e)$ & $\frac{2^{4e} -1 }{2^e-1}$ & $2 \mu(S)$   \\ 

$13$ & $ \POmegaPlus(2d,3)$, $d \geq 4$ & $3 \nmid | A /S|, (\ast) $  & $\frac{3^{d -1}(3^d-1)}{2}$ & $\frac{(3^{d}-1)(3^{d-1}+1)}{2}$ \\ 

$14$ & $\G(3^e)$, $ e >1$ & $A \not \leq \mathrm{\Gamma G}_2(3^e)$ & $\frac{(q^{6}-1)}{q-1}$ & $2 \mu(S)$\\ 

$15$ & $\F(2^e) $ & $A \not \leq \mathrm{\Gamma F}_4(2^e)$ & $\frac{(q^{12}-1)(q^{4}+1)}{q-1}$  & $2 \mu(S)$\\ 

$16$ & $\Esix(q)$ & $A \not \leq \mathrm{\Gamma E}_6(q)$ & $\frac{(q^9-1)(q^{8}+q^{4}+1)}{q-1}$ &$2 \mu(S)$ \\ 
\bottomrule
\end{tabular}}

The condition $(\ast)$ in the row $13$ is $A \not\leq \mathrm{PO}^+(2d,3)$ when $d>4$, and $A$ is not contained in any conjugate of $\mathrm{PO}^+(8,3)$ in $\mathrm{Aut(P\Omega^+}(8,3))$ when $d=4$ \tablefootnote{the group $\rm{PO}^{+}$ is denoted as $\rm{PGO}^{+}$ in the paper \cite{CannonHoltUnger}} .  
\end{table}
\end{proposition}

\subsection{Background on Representation Theory}
A \emph{representation} of a group $G$ is a homomorphism $\eta:G \rightarrow {\rm{GL}}(V)$ for some finite-dimensional non-zero vector space $V$. A subspace $W$ of $V$ is called $G$-\emph{invariant} if for all $g \in G$ and $w \in W$, we have $\eta_g(w) \in W$. A representation $\eta:G \rightarrow {\rm{GL}}(V)$ is called \emph{irreducible} if the only $G$-invariant subspaces of $V$ are ${0}$ and $V$. Let $V$ and $W$ be two finite-dimensional vector spaces over a field $\mathbb{F}$. 

\begin{definition}\label{def-equivalentOfRep.}
Let $G$ be a group. We say that two representations $\vartheta: G \rightarrow {\rm{GL}}(V)$ and $\psi:G \rightarrow {\rm{GL}}(W)$ are \emph{equivalent} if there is an isomorphism $f:V \rightarrow W$ such that  $f\circ \vartheta_g = \psi_g \circ f$ for all $g \in G$. We write $\vartheta \sim \psi$.
\end{definition}

\begin{lemma}[Schur's Lemma o{\cite{serre1977}}]\label{thm-schur} Let $\vartheta: G \rightarrow {\rm{GL}}(V)$ and $\psi:G \rightarrow {\rm{GL}}(W)$ be irreducible representations of a group $G$. Let $f: V \rightarrow W$ be a linear map such that $f \circ \vartheta=\psi \circ f$, then either $f$ is invertible or $f = 0$. Moreover, the only non-trivial linear maps $f$ are identity or scalar multiples of identity.
\end{lemma}

\begin{lemma}\label{lem: Irreducibility}
Let $G$ be one of the following: $\SL(d,q)$, or ${\rm{Sp}}(4, 2^{e})$ ($e \geq 2$). Let $\vartheta: G \rightarrow {\rm{GL}}(d,q)$ be a non-trivial representation of $G$ over a vector space $\mathbb{F}_q^d$ then $\vartheta$ is irreducible.
\end{lemma}

\begin{proof}
Suppose that $\vartheta$ is not irreducible. Let $W \neq 0$ be a proper $G$-invariant subspace of $\mathbb{F}_q^d$. Consider $0 \neq w \in W \subsetneq \mathbb{F}_q^d$. Note that $\{w\}$ can be extended to a basis of $\mathbb{F}_q^d$ , say $\beta := \{w,b_2,\ldots,b_d\}$. Consider a linear transformation, say $T: \mathbb{F}_q^d \rightarrow \mathbb{F}_q^d$ defined as follows: $T(w)=b_2$, $T(b_2)=-w$, and $T(b_i)=b_i$ for all $i \geq 3$. It is easy to see that $\det(T)=1$ and hence $T \in \mathrm{SL}(d,q)$. Now, as $W$ is a $G$-invariant subspace of  $\mathbb{F}_q^d$, we have $T \cdot w =b_2 \in W$. Similarly, we can show that $b_i \in W$ for each $i \geq 3$. This contradicts that $W$ is a proper subspace of  $\mathbb{F}_q^d$. Thus, $\vartheta$ is irreducible when $G=\mathrm{SL}(d,q)$.

For $G=\mathrm{Sp}(d,q)$, where $d$ is even. Observe that the linear transformation $T$ defined above satisfies  $T^{t}XT = X$, where $X$ is the following matrix:
\[X= \begin{bmatrix}
    0 & 1 & & & & & & & \\
    -1 & 0 & &  & &   & \bf{0} & &  \\
   & & 0 & 1  & &  & & &\\
 & & -1 & 0 & &  & & & \\
   & \bf{0}  & & & &\ddots & & & \\
   & & & & & & & 0 & 1   \\
   & & & & & & & -1 & 0   \\
\end{bmatrix}.\] 
Thus, $T \in \mathrm{Sp}(d,q)$ (cf. \Cref{subsec: simplegroup}). In particular, when $d=4$ and $q=2^e$, we get $T \in  \mathrm{Sp}(4,2^e)$. 

Now, as $W$ is a $G$-invariant subspace of  $\mathbb{F}_q^d$, we have $T \cdot w =b_2 \in W$. Similarly, we can show that $b_i \in W$ for each $i \geq 3$. This contradicts the assumption that $W$ is a proper subspace of $\mathbb{F}_q^d$. Thus, $\vartheta$ is irreducible when $G=\mathrm{Sp}(4,2^e)$.
\end{proof}

\begin{lemma}[cf. {\cite[Proposition~2.10.6]{KleidmanLiebeck}}] \label{lem:OmegaIrreducible}
Let $d \geq 2$, and let $q = p^{k}$ for some prime $p > 2$ and some $k \geq 1$. For $U \in \Omega^{+}(2d,q)$ and $\alpha \in \Aut(\Omega^{+}(2d,q))$, define $\psi_{U}^{(\alpha)} : \mathbb{F}_{q}^{2d} \to \mathbb{F}_{q}^{2d}$ by $\psi_{U}^{(\alpha)}(v) = \alpha(U) \cdot v$.

Define the representation $\psi^{(\alpha)} : \Omega^{+}(2d,q) \to \text{GL}(2d,q)$ by $\psi^{(\alpha)}(U) = \psi_{U}^{(\alpha)}$. We have that $\psi^{(\alpha)}$ is irreducible.
\end{lemma}

\begin{proof}
Note that $\psi^{(\text{id})}$ is the natural representation of $\Omega^{+}(2d,q)$. This is irreducible by \cite[Proposition~2.10.6]{KleidmanLiebeck}. Now for any $\alpha \in \Aut(\Omega^{+}(2d,q))$, we have that $\text{Im}(\psi^{(\text{id})}) = \text{Im}(\psi^{(\alpha)})$. The result now follows.
\end{proof}

\section{Computing the Socle of a Fitting-free Group}

In this section, we will establish the following.

\begin{proposition} \label{prop:ComputeSocleFittingFree}
Let $G \leq \Sym(n)$ be a Fitting-free permutation group. We can compute generators for $\Soc(G)$ in $\textsf{NC}$. Furthermore, we can in $\textsf{NC}$ decompose $\Soc(G)$ into a direct product of non-abelian simple groups.
\end{proposition}

We first recall the following standard lemma.

\begin{lemma}[cf. \cite{Holt2005HandbookOC}] \label{lem:HandbookSocle}
Let $G$ be Fitting-free. We have that for any $N \trianglelefteq G$, $\Soc(G) = \Soc(N) \times \Soc(C_{G}(\Soc(N)))$.   
\end{lemma}

We are also grateful to Joshua A. Grochow \cite{GrochowLemmas} for sharing the following key lemmas and their proofs.

\begin{lemma}[{\cite{GrochowLemmas}}] \label{lem:normal-subgroup-FF}
If $G$ is Fitting-free and $N \unlhd G$, then $N$ is Fitting-free and $\Soc(N) = \Soc(G) \cap N$.
\end{lemma}

\begin{proof}
Every subgroup of $N$ that is normal in $G$ is also normal in $N$, so $\Soc(N)$ cannot contain any non-minimal normal subgroups of $G$. And since $\Soc(N)$ is characteristic in $N$, it is normal in $G$, and thus $\Soc(N) \subseteq \Soc(G) \cap N$.

In the opposite direction, suppose $S \leq N \cap \Soc(G)$ is a minimal normal subgroup of $G$. Then $S$ is a direct product of non-abelian simple groups since $G$ is Fitting-free. As $S$ is normal in $N$ and a direct product of non-abelian simple groups, it is contained in $\Soc(N)$ (see also \cite[Lemma~4.3]{GLDescriptiveComplexity}). Thus $N \cap \Soc(G) \leq \Soc(N)$, so the two are equal. Finally, since $\Soc(N)$ is a normal subgroup of $\Soc(G)$, and the latter is a direct product of non-abelian simple groups, so is the former. But then $N$ contains no abelian normal subgroups, so is Fitting-free.
\end{proof}

\begin{lemma}[\cite{GrochowLemmas}]\label{lem:socle}
Suppose $G$ is Fitting-free and $1 \leq N_1 \leq \dotsb \leq N_r = G$ is a normal series (all $N_i$ are normal in $G$) such that each factor $N_{i+1} / N_{i}$ is a direct product of simple groups, and is either abelian, or contains no abelian direct factors. Define inductively $M_1 := N_1$, $M_i := M_{i-1} \times (N_i \cap C_G(M_{i-1}))$. Then for all $i$, we have $M_i = \Soc(N_i)$, and in particular, $M_r = \Soc(G)$.
\end{lemma}

\begin{proof}
By induction on $i$. For $i=1$: Since $G$ is Fitting-free and $N_1 = M_1$ is a normal subgroup, it cannot be abelian, so by assumption it is a direct product of non-abelian simple groups and thus $M_1 = N_1 = \Soc(N_1)$. 

Now suppose that $M_i = \Soc(N_i)$ for some $i \geq 1$; we will show the same holds for $i+1$.

\textbf{Case 1:} $\Soc(N_i) = \Soc(N_{i+1})$. By induction we have $M_i = \Soc(N_i)$. Since $N_{i+1}$ is Fitting-free, we have that $1 = C_{N_{i+1}}(\Soc(N_{i+1}))$. But the latter is the same as $C_G(\Soc(N_{i+1})) \cap N_{i+1} = C_G(\Soc(N_i)) \cap N_{i+1} = C_G(M_i) \cap N_{i+1}$. Thus $M_{i+1} = M_i = \Soc(N_{i+1})$.

\textbf{Case 2:} $\Soc(N_{i+1}) \neq \Soc(N_i)$. By Lemma~\ref{lem:normal-subgroup-FF}, $N_i \cap \Soc(G) = \Soc(N_i)$ and $N_{i+1} \cap \Soc(G) = \Soc(N_{i+1})$, thus $\Soc(N_i) \lneq \Soc(N_{i+1})$. Since $\Soc(N_i)$ is characteristic in $N_i$ which is normal in $G$, we have that $\Soc(N_i)$ is also normal in $G$. Hence $\Soc(N_i)$ is also normal in $\Soc(N_{i+1})$. Since all these socles are direct products of non-abelian simple groups, there is some $S \leq G$ that is a direct product of non-abelian simple groups such that $\Soc(N_{i+1}) = \Soc(N_i) \times S$. By the definition of $M_{i+1}$, it is then necessary and sufficient to show $N_{i+1} \cap C_G(M_i) = S$.

As $S$ centralizes $\Soc(N_i) = M_i$, we have that $S$ is contained in both $N_{i+1}$ and $C_G(M_i)$. Therefore $S \leq N_{i+1} \cap C_G(M_i)$. 

In the opposite direction, suppose $g \in N_{i+1} \cap C_G(M_i)$; we will show that $g$ must be in $S$. Since $C_{N_{i+1}}(\Soc(N_{i+1})) = 1$, $g$ must not centralize all of $\Soc(N_{i+1})$. But since $g$ centralizes $M_i = \Soc(N_i)$, $g$ must not centralize $S$. Let $N$ be the normal closure of $g$ in $G$. Since $N_{i+1} \unlhd G$, we have $N \leq N_{i+1}$ and hence $N \unlhd N_{i+1}$. Additionally, since $M_i = \Soc(N_i)$, we have argued above that $M_i \unlhd G$, and since the centralizer of a normal subgroup is normal, $C_G(M_i)$ is also normal in $G$, and thus $N \leq C_G(M_i)$. Since $N$ is normal in $G$, it is Fitting-free, and we have $\Soc(N) = N \cap \Soc(G) \leq N_{i+1} \cap \Soc(G) = \Soc(N_{i+1})$. Since $N \leq C_G(M_i) = C_G(\Soc(N_i))$, we must have $\Soc(N) \cap \Soc(N_i) = 1$. But $\Soc(N) \leq \Soc(N_{i+1}) = \Soc(N_i) \times S$, and $\Soc(N)$ is normal in $N_{i+1}$, so---since these socles are all direct products of non-abelian simple groups---we must have $\Soc(N) \leq S$, and $g$ does not centralize $\Soc(N)$. But if $g$ is not in $\Soc(N)$, then $N$ is not a direct product of non-abelian simple groups so we must in fact have $g \in \Soc(N) \leq S$, as claimed.

Therefore $S = N_{i+1} \cap C_G(M_i)$, and $M_{i+1} = M_i \times S = \Soc(N_{i+1})$. \end{proof}

\begin{proof}[Proof of Proposition~\ref{prop:ComputeSocleFittingFree}]
Babai, Luks, and Seress \cite{BabaiLuksSeress} established that there exists an $\textsf{NC}$-computable series
\[
1 = N_{0} \leq N_{1} \leq \ldots \leq N_{r} = G
\]
of length $r \in O(\log^{c} n)$ (Kantor, Luks, and Mark \cite{KantorLuksMark} determined that $c = 2$) such that for all $i \in [r]$, $N_{i} \trianglelefteq G$ (see \cite[Sec.~3.6]{KantorLuksMark}) and each $N_{i}/N_{i-1}$ is a direct product of simple groups (semisimple). In particular, $N_{i}/N_{i-1}$ is either abelian or has no abelian direct factor. We crucially take advantage of both the length of the series, as well as the fact that the successive quotients are semisimple to compute $\Soc(G)$ in $\textsf{NC}$.

We proceed inductively along the series. We begin with $N_{1}$. As $G$ is Fitting-free, $N_{1}$ is the direct product of non-abelian simple groups. By Lemma~\ref{lem:normal-subgroup-FF}, $N_{1} = \Soc(N_1) \trianglelefteq \Soc(G)$. By Lemma~\ref{lem:HandbookSocle}, $\Soc(G) = N_1 \times \Soc(C_{G}(N_1))$.

Define $M_{1} := N_{1} (= \Soc(N_1))$. And for $i > 1$, define $M_{i} := M_{i-1} \times (N_{i} \cap C_{G}(M_{i-1}))$. As each $N_{i}$ is normal in $G$, we have that each $M_{i}$ is normal in $G$. Now by Lemma~\ref{PermutationGroupsNC}(f), we may compute $C_{G}(M_{i-1})$ in $\textsf{NC}$, provided we are given generators for $M_{i-1}$. By Lemma~\ref{PermutationGroupsNC}(g), we may then compute $C_{G}(M_{i-1}) \cap N_{i}$ in $\textsf{NC}$. Thus, we may compute $M_{i}$ in $\textsf{NC}$, provided we are given generators for $M_{i-1}$. As $r \in O(\log^{2} n)$, we may compute $M_{1}, \ldots, M_{r}$ in $\textsf{NC}$. By Lemma~\ref{lem:socle}, we have that $M_{r} := \Soc(G)$.

Once we have generators for $\Soc(G)$, we may decompose $\Soc(G)$ into a direct product of non-abelian simple groups in $\textsf{NC}$ using \cite{BabaiLuksSeress} (see also the discussion following the proof of \cite[Proposition~2.1]{BCGQ}). The result now follows.
\end{proof}

\begin{remark}
In the setting of both permutation groups and their quotients, the socle is polynomial-time computable \cite{KantorLuksQuotients}. However, obtaining $\textsf{NC}$ bounds is a long-standing open problem. Extending Proposition~\ref{prop:ComputeSocleFittingFree} to the setting of quotients is also non-trivial. We rely crucially on the $\textsf{NC}$ algorithm of Babai, Luks, and Seress \cite{BabaiLuksSeress} to compute the centralizer of a normal subgroup $H$ in $\textsf{NC}$. Kantor and Luks \cite[Section~6]{KantorLuksQuotients} note that the techniques for computing $C_{G}(H)$ in the setting of permutation groups appear not to work for  quotients. Instead, Kantor and Luks bring to bear the extensive machinery for computing Sylow subgroups of permutation groups in order to compute cores of quotients. Kantor and Luks then reduce computing $C_{G}(H)$ to computing an appropriate core. It is unclear how to compute cores of quotients in $\textsf{NC}$.
\end{remark}

\begin{proposition} \label{prop:MinimalNormalFittingFree}
Let $K \trianglelefteq G \leq \Sym(n)$, and let $\textbf{G} := G/K$ be Fitting-free. Suppose we are given the non-abelian simple factors of $\Soc(G)$. Then in $\textsf{NC}$, we can compute the set of minimal normal subgroups of $\textbf{G}$.
\end{proposition}

\begin{proof}
The conjugation action of $\textbf{G}$ on $\Soc(\textbf{G})$ induces a permutation action on the non-abelian factors of $\Soc(\textbf{G})$. There are $\leq \log |G| \leq \log n! \sim n \log n$ such factors. View \textbf{G} as a permutation group $G^{*}$, with the non-abelian simple factors as the permutation domain. For every given orbit $\mathcal{O}$ of $G^{*}$, the members of $\mathcal{O}$ are precisely the factors of a single non-abelian simple factor of a minimal normal of $\textbf{G}$. Furthermore, all such minimal normal subgroups of $\textbf{G}$ are realized in this way. Now computing orbits in a permutation group is $\textsf{NC}$-computable \cite{McKenzieThesis}. The result now follows.
\end{proof}

Finally, we will show that if we are given a non-abelian simple factor $S$ is a non-abelian simple factor of the socle, for a Fitting-free group $\textbf{G}$, we can compute $N_{\textbf{G}}(S)$ in $\textsf{NC}$.

\begin{proposition}\label{Normalizer}
Let $G, K \leq \Sym(n)$ with $K \trianglelefteq G$. Let $\textbf{G} = G/K$. Suppose that $\textbf{G}$ is Fitting-free, and we are given the non-abelian simple factors of $\Soc(\textbf{G})$. Let $S$ be such a simple factor. Computing $N_{\textbf{G}}(S)$ is $\textsf{NC}$-reducible to \algprobm{Pointwise Stabilizer}. In particular, we can compute $N_{\textbf{G}}(S)$ in $\textsf{NC}$.
\end{proposition}

\begin{proof}
The conjugation action of $\textbf{G}$ on $\Soc(\textbf{G})$ induces a permutation action on the non-abelian simple factors of $\Soc(\textbf{G})$. There are $\leq \log |\textbf{G}| \leq \log n! \sim n \log n$ such factors. View $\textbf{G}$ as a permutation group $G^{*}$, with the non-abelian simple factors as the permutation domain. 

We claim that, in $\textsf{NC}$, we can decide whether an element of $\textbf{G}$ fixes $S$ setwise. Given generators for $S$ and some $g \in \textbf{G}$, we can use \algprobm{Membership} to test whether $\langle S \rangle = \langle S \rangle^{g}$. By Lemma~\ref{PermutationGroupsNC}(b), \algprobm{Membership} is $\textsf{NC}$-computable. Thus, we can decide in $\textsf{NC}$ whether an element of $G^{*}$ stabilizes $S$ as a point. The point stabilizer $\text{Stab}_{G^{*}}(S) = N_{\textbf{G}}(S)$. By Lemma~\ref{PermutationGroupsNC}(d), \algprobm{Point Stabilizer} is $\textsf{NC}$-computable. The result follows.
\end{proof}

\section{Lifting of Automorphisms}

\subsection{Lifting Automorphisms of $\PSL(d,q)$ to Automorphisms of $\text{SL}(d,q)$}
\label{sec:LiftSL}

In this section, we extend the work of \cite[Section~3]{DasThakkarMPD} to show that we can lift an automorphism of $\PSL(d,q)$ to an automorphism of $\text{SL}(d,q)$ in $\textsf{NC}$. Throughout this section, we will fix $n$ and consider quotients of permutation groups of $\Sym(n)$. We will also assume throughout this section that $d \geq 3$, and $(d,q) \neq (3,2), (4,2)$. We begin with the following observation. 

\begin{observation} \label{obs:PSLd}
Let $G, K \leq \Sym(n)$, with $K \trianglelefteq G$. Let $\textbf{G} := G/K$. Suppose that $\PSL(d,q) \leq \textbf{G}$. Then (i) $d \in O(\sqrt{n \log n})$, and (ii) $q \leq n$. 
\end{observation}

\begin{proof}
We note that $|\PSL(d,q)| \in q^{O(d^2)} \leq |\textbf{G}| \leq |G|$. It follows that $d^{2} \log q \leq \log |G| \leq \log n! \leq n \log n$, which establishes (i). However, this only provides the bound that $\log q \leq n \log n$. To see that $q \leq n$, let $q'$ be the largest prime power such that $\text{PSL}(d, q') \leq \Sym(n)$. In this case, $q' \leq n$ \cite{KantorSeress}. Clearly, $q \leq q'$, and so $q \leq n$, as desired.  
\end{proof}

In light of Observation~\ref{obs:PSLd}, we may assume without loss of generality that $d, q$ are given in unary. Throughout this section, let $Z \leq \text{SL}(d, q)$ denote the subgroup of all nonzero scalar transformations of $\mathbb{F}_{q}^{d}$. Note that $Z = Z(\text{SL}(d,q))$. 

\begin{proposition}[{\cite[Proposition~3.1]{DasThakkarMPD}}] \label{prop-UniEle}
Let $d \geq 3$, $p$ be prime, and $e$ such that $q = p^{e}$. Suppose that $(d, q) \not \in \{(3,2), (4,2)\}$. Each coset of $\text{SL}(d,q)/Z$ has at most one element of order $p$. 
\end{proposition}

\begin{theorem}[\cite{Dieudonne}]
$\Aut(\PSL(d,q)) \cong \Aut(\text{SL}(d,q))$. 
\end{theorem}

Let $\Psi : \Aut(\text{SL}(d,q)) \to \Aut(\PSL(d,q))$ be the map that takes $\alpha \in \Aut(\text{SL}(d,q))$ to the map $\overline{\alpha} : \PSL(d,q) \to \PSL(d,q)$, where $\overline{\alpha}(UZ) = \alpha(U)Z$. Das and Thakkar established \cite[Proposition~3.3]{DasThakkarMPD} that $\Psi$ is indeed an isomorphism. 

Theorem~\ref{thm:ConstructiveRecognition}  provides a framework to convert between the matrix and permutation representations. We take advantage of this in tandem with a prescribed generating set of matrices \cite[Page~185]{carter-book} to lift automorphisms of $\PSL(d,q)$ to automorphisms of $\text{SL}(d,q)$.

\begin{lemma}[cf. {\cite[Lemma~3.4]{DasThakkarMPD}}]\label{construct-alpha}
Given $\lambda \in \Aut(\PSL(d,q))$ by its action on a generating set of $\text{PSL}(d, q)$, we can find $\alpha \in \Aut(\text{SL}(d,q))$ such that $\lambda = \overline{\alpha}$. Moreover, $\alpha$ is specified by its action on a generating set of size $O(qd^{2})$. 

If $\PSL(d,q)$ is given as a quotient of permutation groups $G/K$, where $K \trianglelefteq G \leq \Sym(n)$, our procedure is $\textsf{NC}$-computable. 
\end{lemma}

\begin{proof}
By Lemma~\ref{QuotientsNC}, we can compute $|\PSL(d,q)|$ in $\textsf{NC}$. As $q$ is given in unary, we may write down $\mathbb{F}_{q}$ in $\textsf{NC}$ (by writing down the Cayley tables for $(\mathbb{F}_{q}, +)$ and $(\mathbb{F}_{q}^{*}, \times)$). For $i, j \in [d]$, let $e_{i,j}$ be the $d \times d$ matrix with $1$ in the $(i,j)$ entry and $0$ otherwise. Furthermore, in $\textsf{NC}$, we may write down the following set:
\[
L = \{ I + \beta e_{i,j} \mid \beta \in \mathbb{F}_{q}^{*}; i, j \in [d] \text{ s.t. } i \neq j \}.
\]
Note that $\text{SL}(d,q) = \langle L \rangle$ (see e.g., \cite[Page~185]{carter-book}).

For each $U \in L$, $\text{ord}(U) = p$, which is coprime to $|Z| = \text{gcd}(d, q-1)$. Let $m_{2} : \text{SL}(d,q) \to \text{PSL}(d, q)$ be the canonical projection map, with kernel $Z$. We have the following commutative diagram: 

$$ \begin{tikzcd}
\mathrm{SL}(d,q) \arrow{r}{\alpha} \arrow[swap]{d}{m_2} & \mathrm{SL}(d,q) \arrow{d}{m_2} \\
\PSL(d,q) \arrow{r}{\lambda}& \PSL(d,q)
\end{tikzcd}.$$ 

From the above diagram, we have that $m_{2} \circ \alpha = \lambda \circ m_{2}$. Now observe that if $U \in L$, then $\text{ord}(\alpha(U)) = \text{ord}(U) = p$. Furthermore, observe that there exists some $V \in \text{SL}(d,q)$ such that $\lambda(m_{2}(U)) = VZ$.

We now need to define the image $\alpha(U)$ for each $U \in L$. As $\text{ord}(U) = p$ and $\alpha \in \Aut(\text{SL}(d,q))$, we have that $\text{ord}(\alpha(U)) = p$. As the diagram is commutative, $\alpha(U)$ is an element of the coset $VZ$. So we have that $\alpha(U), m_{2}(\alpha(U)) \in VZ$. However, by \cite[Proposition~3.1]{DasThakkarMPD}, there is at most one element of $VZ$ of order $p$. Hence, $\alpha(U) = m_{2}(\alpha(U))$. As $|Z| = \text{gcd}(d, q-1) \leq n$, we can in $\textsf{NC}$, using Lemma~\ref{lem:List}, list $Z$ and hence $VZ$. Thus, we can in $\textsf{NC}$ find $\alpha(U) = m_{2}(\alpha(U))$. 

Now, for each $U \in L$, we can in $\textsf{NC}$ construct a word over the specified generating set of $\PSL(d,q)$. Here, we use the fact that $\PSL(d,q)$ is given as a quotient of permutation groups. Thus, by Lemma~\ref{QuotientsNC}, we can evaluate $\lambda(m_{2}(U))$ in $\textsf{NC}$. The corresponding word, in particular, yields a coset representative of $VZ$. As we have written down $Z$, we can identify $V$ in $\textsf{NC}$. The result now follows.
\end{proof}

We now turn to identifying the type of a given automorphism $\alpha$ of $\text{SL}(d,q)$. We follow the strategy of Das and Thakkar \cite{DasThakkarMPD}. Let $G = \text{SL}(d,q)$ and $V = \mathbb{F}_{q}^{d}$, where $q = p^{e}$ for some prime $p$ and some $e > 0$. Consider a representation $\vartheta : \text{SL}(d,q) \to \text{GL}(d,q)$, which takes $U$ to $\vartheta_{U}$, where $\vartheta_{U} : V \to V$ is defined by $\vartheta_{U}(v) = U \cdot v$. Similarly, consider the representation $\psi^{(\alpha)} : \text{SL}(d,q) \to \text{GL}(d,q)$ which takes $U \in \text{SL}(d,q)$ to $\psi_{U}^{(\alpha)} : V \to V$ defined by $\psi_{U}^{(\alpha)}(v) = \alpha(U) \cdot v$. Note that we can define $\vartheta, \psi^{(\alpha)}$ analogously for ${\rm{Sp}}(4, 2^e)$ ($e \geq 2$). By Lemma~\ref{lem: Irreducibility}, both $\vartheta_{U}$ and $\psi^{(\alpha)}$ are irreducible.

Das and Thakkar \cite[Remark~3.5]{DasThakkarMPD} noted that $\vartheta \sim \psi^{(\alpha)}$ if and only if there exists $F \in \text{GL}(d, q)$ such that for all $U \in \text{SL}(d,q)$, we have that $FUF^{-1} = \alpha(U)$. 

\begin{lemma}[cf. {\cite[Lemma~3.6]{DasThakkarMPD}}]\label{lem-type-checking}
Let $G$ be one of the following: $\SL(d,q)$ ($(d, q) \not \in \{ (3,2), (4,2)\})$, ${\rm{Sp}}(4, 2^e)$ ($e \geq 2$). Let $\alpha$ be an automorphism of $G$. Suppose that $\alpha$ is specified by its action on a generating set $B$ of $G$, consisting of $d \times d$ matrices over $\mathbb{F}_{q}$ in $G$. Furthermore, suppose that $\vartheta$ and $\psi^{(\alpha)}$ are irreducible. Then in $\textsf{NC}^{2}$, we can decide if $\vartheta \sim \psi^{(\alpha)}$. 
\end{lemma}

\begin{proof}[Proof of Lemma~\ref{lem-type-checking}]
We follow the strategy of \cite[Lemma~3.6]{DasThakkarMPD}. Our goal is to find an $F \in \text{GL}(d, q)$ such that for all $U \in G$, we have that $FUF^{-1} = \alpha(U)$. Let $B = \{ U_1, \ldots, U_b\}$ be a generating set for $G$. Consider the following system of linear equations:
\[
FU_{j} = \alpha(U_{j})F, \text{ for all } j \in [b],
\]

\noindent where the entries of $F$ are the unknowns. This is a homogeneous system of $bd^{2}$ equations with $d^{2}$ unknowns. We can solve this system of equations in $\textsf{NC}^{2}$ using Gaussian elimination \cite{MulmuleyRank}. If there is a unique $F = 0$, then there is no invertible matrix $F$ satisfying $FU = \alpha(U)F$ for all $U \in G$. In this case, $\vartheta \not \sim \psi^{(\alpha)}$. Otherwise, there is a non-zero solution satisfying $F\alpha(U) = UF$. But by Lemma~\ref{thm-schur}, we have that $F$ is invertible. So $\vartheta \sim \psi^{(\alpha)}$. The result now follows.
\end{proof}

\begin{observation}[{\cite[Observation~3.7]{DasThakkarMPD}}]\label{Observation-typechecking-PSL}
Let $G$ be one of the following: $\SL(d,q)$, ${\rm{Sp}}(4, 2^e)$ ($e \geq 2$). Let $\alpha \in \aut(G)$. Then the following are equivalent:
\begin{enumerate}[label=(\alph*)]
\item The irreducible representations $\vartheta, \psi^{(\alpha)}$ are equivalent,
\item there exists a diagonal matrix $F \in \mathrm{GL}(d,q)$ such that $\alpha(U)=F^{-1}UF$ for all $U \in G$, 
\item $\alpha$ is a diagonal automorphism of $G$. Moreover, if a diagonal matrix $F \in \mathrm{GL}(d,q)$ also belongs to $G$ then $\alpha$ is also an inner automorphism of $G$.
\end{enumerate}
\end{observation}

\begin{proof}
We note that, if (a) holds, i.e., $\vartheta \sim \psi^{(\alpha)}$, which implies that there exists $F \in \mathrm{GL}(d,q)$ such that $F \vartheta = \psi^{(\alpha)}F$, i.e., for all $U \in G$ we have $F U= \alpha(U)F$. By \Cref{thm-schur}, such $F$ is always diagonal, which proves (b), and also yields the implication $(b) \implies (c)$. If $\alpha$ is a diagonal automorphism of $G$, then clearly $\vartheta \sim \psi^{(\alpha)}$.
\end{proof}

\subsection{Lifting Automorphisms of $\POmegaPlus(2d,q)$ to Automorphisms of $\Omega^+(2d,q)$}

In this section, we extend the work of \cite[Section~7]{DasThakkarMPD} to show that we can lift an automorphism of of $\POmegaPlus(2d, q)$ to an automorphism of $\Omega^{+}(2d,q)$, in $\textsf{NC}$. Throughout this section, we will fix $n$ and consider quotients of permutation groups of $\text{Sym}(n)$. Additionally, throughout this section, we will assume that $q \geq 3$. We begin with the following observation.

\begin{observation} \label{obs:POmega}
Let $G, K \leq \text{Sym}(n)$, with $K \trianglelefteq G$. Let $\textbf{G} := G/K$. Suppose that $\POmegaPlus(2d,q) \leq \textbf{G}$. Then (i) $d \in O(\sqrt[3]{n \log n})$, and (ii) $q \leq n$.     
\end{observation}

\begin{proof}
Note that $|\Omega^{+}(2d, q)| \leq q^{c \cdot d^3}$ for some $c \geq 1$. Thus, $d^{3} \log q \leq \log |G| \leq \log n! \leq n \log n$. Thus, $d \in O(\sqrt[3]{n \log n})$. This establishes (i).

To see that $q \leq n$, let $q'$ be the largest prime power such that $\POmegaPlus(2d,q) \leq \text{Sym}(n)$. In this case, $q' \leq n$ \cite{KantorSeress}. Clearly, $q \leq q'$, and so $q \leq n$ as desired.
\end{proof}

In light of Observation~\ref{obs:POmega}, we may assume without loss of generality that $d, q$ are given in unary.

\begin{proposition}[{\cite[Proposition~7.1]{DasThakkarMPD}}]\label{prop-UniEle-POmega}
Let $d \geq 4$, and let $q = p^{e}$ where $p$ is prime. Then each coset of $Z=Z(\rm{O}^+(2d,q)) \cap \Omega^+(2d,q)$ in $\Omega^+(2d,q)$ has at most one element of order $p$.
\end{proposition}

\begin{proof}
Let $U, U\Lambda \in UZ$ have order $p$. If $q$ is even, then $p = 2$. In this case, $Z(\rm{O}^{+}(2d,q)) = \{ I_{2d}\}$, in which case $Z = \{I_{2d}\}$ and we are done. So suppose that $q$ is odd. In this case, $Z = \{ \pm I_{2d} \}$. Here, we have that $\text{ord}(U) = p$ and $\text{ord}(\Lambda) \leq 2$. As $U, \Lambda$ commute, we have that $p = \text{ord}(U\Lambda) = \text{ord}(U) \text{ord}(\Lambda)$. As $\text{ord}(U) = p$, we thus have that $\text{ord}(\Lambda) = 1$. So $\Lambda = I_{2d}$. The result follows.     
\end{proof}

\begin{remark}
Note that $Z(O^{+}(2d,q)) = \{ I_{2d}, - I_{2d} \}$ (if $q$ is even, then $I_{2d} = -I_{2d}$). It follows that $Z=Z(\rm{O}^+(2d,q)) \cap \Omega^+(2d,q) = \{ I_{2d}, -I_{2d} \}.$ \cite{KleidmanLiebeck}.
\end{remark}

\begin{theorem}[cf. {\cite[Theorem 4.5]{DasThakkarMPD}}]
\noindent 
\begin{enumerate}[label={(\roman*)},itemindent=1em]\label{Thm-Iso-POmegga-Omega}
\item[(i)] \cite{dieudonne1951} $\aut(\POmegaPlus(2d,q)) \cong \aut(\Omega^+(2d,q))$.

\item[(ii)] Let $\Psi:\aut(\Omega^+(2d,q))  \rightarrow \aut(\POmegaPlus(2d,q))$ be the map that takes $\alpha \in \aut(\Omega^+(2d,q))$ to $\bar{\alpha}$, where $\bar{\alpha}:\POmegaPlus(2d,q) \rightarrow \POmegaPlus(2d,q)$ is defined as $\bar{\alpha}(UZ)=\alpha(U)Z$. Then the map $\Psi:\aut(\Omega^+(2d,q)) \rightarrow \aut(\POmegaPlus(2d,q))$ is an isomorphism. 
\end{enumerate}
\end{theorem}

The proof of \Cref{Thm-Iso-POmegga-Omega}(ii) is similar to the proof of \cite[Proposition~3.3]{DasThakkarMPD} with $L$ being a generating set of $\Omega^{+}(2d,q)$ defined below in the proof of \Cref{construct-alpha-POmega}. \Cref{Thm-Iso-POmegga-Omega} says that for any $ \lambda \in \aut(\POmegaPlus(2d,q))$ there is $\alpha \in \aut(\Omega^+(2d,q))$ such that $\lambda=\bar{\alpha}$. We now describe how to compute such elements $\alpha$ from a given $\lambda$.

In a similar manner as Section~\ref{sec:LiftSL}, we will  leverage Theorem~\ref{thm:ConstructiveRecognition}, which provides a framework to convert between the matrix and permutation representations. We take advantage of this in tandem with a prescribed generating set of matrices \cite[Page~185]{carter-book} to lift automorphisms of $\POmegaPlus(2d,q)$ to automorphisms of $\Omega^+(2d,q)$.

\begin{lemma}\label{construct-alpha-POmega}
Let $q = p^{e}$ for some prime $p$ and some $e > 0$. Given $\lambda \in \aut(\POmegaPlus(2d,q))$ by its action on a generating set of $\POmegaPlus(2d,q)$, we can find $\alpha\in \aut(\Omega^+(2d,q))$ such that $\lambda=\bar{\alpha}$. Moreover, $\alpha$ is specified by its action on a generating set of size $O(d^2)$.

If $\POmegaPlus(2d,q)$ is given as a quotient $G/K$ of permutation groups, where $K \trianglelefteq G \leq \Sym(n)$, then our procedure is $\textsf{NC}$-computable.
\end{lemma}

\begin{proof}
Note that if $q$ has even characteristic, then $Z = \{ I_{2d}\}$. In this case, $\Omega^{+}(2d,q) = \POmegaPlus(2d,q)$. So we may take $\alpha = \lambda$, and we are done. So for the remainder of this proof, assume that $q$ has odd characteristic. For a matrix $U \in \mathrm{\Omega^{+}}(2d,q)$, numbering the row and columns by $1,\ldots, d, -1, \ldots, -d$. For $i, j \in [d] \cup \{-1, \ldots, -d\}$, let $e_{i,j}$ be the $2d \times 2d$ matrix with $1$ in the $(i,j)$ entry and $0$ otherwise. Furthermore, in $\textsf{NC}$, we may write down the following set:
\begin{align*}
L=\{&I+\beta(e_{i,j}-e_{-j,-i}),\\
&I-\beta(e_{-i,-j}-e_{j,i}),\\ 
&I+\beta(e_{i,-j}-e_{j,-i}),\\ 
&I-\beta(e_{-i,j}-e_{-j,i})  \mid 0 < i < j, \beta \in \mathbb{F}_3\}.
\end{align*}

Note that $\mathrm{\Omega^{+}}(2d,q)=\langle L \rangle$ (see e.g., \cite[Page~185]{carter-book}). We now claim that for each $U \in L$, $\text{ord}(U)$ is coprime to $|Z|$. As $q$ has odd characteristic, we have that $Z = \{ \pm I_{2d}\}$ and $\text{ord}(U) = p$. Thus, $\text{ord}(U)$ is coprime to $|Z|$. Let $m_2: \Omega^+(2d,q) \rightarrow \POmegaPlus(2d,q)$ be the canonical map with kernel $Z$. We have the following commutative diagram.

$$ \begin{tikzcd}
\Omega^+(2d,q) \arrow{r}{\alpha} \arrow[swap]{d}{m_2} & \Omega^+(2d,q) \arrow{d}{m_2} \\
\POmegaPlus(2d,q) \arrow{r}{\lambda}& \POmegaPlus(2d,q)
\end{tikzcd}.$$ 

From the above diagram, we have $m_2 \circ {\alpha} = \lambda \circ m_2$. Let $U \in L$ then $\ord(\alpha(U))=\ord(U)=p$. Notice that $\lambda(m_2 (U))= VZ$ for some $V \in \Omega^+(2d,q)$ \cite{KantorSeress}.

We now need to define the image $\alpha(U)$ for each $U \in L$. Since $\ord(U)=p$ and $\alpha \in \aut(\Omega^+(2d,3))$ we have $\ord(\alpha(U))=p$. Since the diagram is commutative, we have $\alpha(U)$ is an element in the coset $VZ$. In other words, $\alpha(U) \in VZ$ and $(m_2 \circ \alpha)(U)=VZ$. But ${\alpha}(U)$ is of order $p$ and there is just one element in $VZ$ with order $p$ (cf. \Cref{prop-UniEle-POmega}). Since $|Z|=2$, we can easily find such an element of order $p$ in the coset $VZ$. Thus, we can in \textsf{NC} find $\alpha(U)$.

Now, for each $U \in L$, we do the following. Here, we use the fact that $\POmega^{+}(2d,q)$ is given as a quotient of permutation groups. Thus, by Lemma~\ref{QuotientsNC}, we can in $\textsf{NC}$, construct a word over the specified generating set of $\POmegaPlus(2d,q)$. Thus, we can in $\textsf{NC}$, evaluate $\lambda(m_{2}(U))$. The corresponding word, in particular, yields a coset representative of $VZ$. As we have written down $Z$, we can identify $V$ in $\textsf{NC}$. The result now follows.
\end{proof}

Now we are interested in identifying the type of automorphism $\alpha$ of $\Omega^+(2d,q)$. Again, we follow the strategy of Das and Thakkar \cite{DasThakkarMPD}. Let $G=\Omega^+(2d,q)$, $V=\mathbb{F}_{q}^{2d}$. Consider a representations $\vartheta: \Omega^+(2d,q) \rightarrow {\rm{GL}}(2d,q)$ which takes ${U} \in \Omega^+(2d,q)$ to $\vartheta_{U}$, where ${\vartheta}_{U}: V \rightarrow V$ defined as ${\vartheta}_U(v)=U\cdot v$. Also, take $\alpha \in \aut(\Omega^+(2d,q))$. Consider $\psi^{(\alpha)}: \Omega^+(2d,q) \rightarrow {\rm{GL}}(2d,q)$ which takes $U \in \Omega^+(2d,q)$ to ${\psi}_U^{(\alpha)}$, where ${\psi}_{U}^{(\alpha)}: V \rightarrow V$ defined as ${\psi}_U^{(\alpha)}(v)=\alpha(U)\cdot v$. By Lemma~\ref{lem:OmegaIrreducible}, both $\vartheta$ and $\psi^{(\alpha)}$ are irreducible representations.

\begin{remark}\label{rem module-equ-cond-POmega}
Let $\vartheta,\psi^{(\alpha)}$ be as defined above. Then $\vartheta \sim \psi^{(\alpha)}$ if and only if there is $F \in \mathrm{GL}(2d,q)$ such that ${FUF^{-1}}=\alpha(U)$ for all $U \in G$. 
\end{remark}

\begin{lemma}\label{lemma-module-equ-cond-POmega}
Let $\vartheta,\psi^{(\alpha)}$ be as defined above. Then $\vartheta \sim \psi^{(\alpha)}$, i.e., there is $F \in \mathrm{GL}(2d,3)$ such that $FU=\alpha(U)F$ for all $U \in G$. Furthermore, we can in $\textsf{NC}^{2}$, decide if $\vartheta \sim \psi^{(\alpha)}$.
\end{lemma}
\begin{proof}    
Suppose $\alpha \in \aut(G)$. Let $B = \{ U_1, \ldots, U_b\}$ be a generating set for $\Omega^+(2d,q)$, where $U_1, \ldots, U_b$ are represented as $2d \times 2d$ matrices over $\mathbb{F}_{q}$. Consider the following system of linear equations:
\[
FU_{j} = \alpha(U_{j})F, \text{ for all } j \in [b],
\]

\noindent where the entries of $F$ are the unknowns. This is a homogeneous system of $bd^{2}$ equations with $d^{2}$ unknowns. We can solve this system of equations in $\textsf{NC}^{2}$ using Gaussian elimination \cite{MulmuleyRank}. Let $F$ be a non-zero solution obtained from the Gaussian elimination algorithm. By \Cref{thm-schur} we know that $F$ is invertible. Thus, the result now follows. 
\end{proof}

\begin{observation}[cf. {\cite[Observation~7.6]{DasThakkarMPD}}]\label{Observation-typechecking-POmega}
Let $\alpha \in \aut(\Omega^+(2d,q))$. The following are equivalent: 
\begin{enumerate}[label=(\alph*)]
\item $\vartheta \sim \psi^{(\alpha)}$,
\item there exists a diagonal matrix  $F \in \mathrm{GL}(2d,q)$ such that $\alpha(U)=F^{-1}UF$ for all $U \in \Omega^+(2d,q)$, and
\item $\alpha$ is a diagonal automorphism of $G$. Moreover, if a diagonal matrix $F \in \mathrm{GL}(d,q)$ also belongs to $\Omega^+(2d,q))$ then $\alpha$ is also an inner automorphism of $G$.
\end{enumerate}
\end{observation}
\begin{proof}
We note that, if (a) holds, i.e., $\vartheta \sim \psi^{(\alpha)}$, which implies that there exists $F \in \mathrm{GL}(d,q)$ such that $F \vartheta = \psi^{(\alpha)}F$, i.e., for all $U \in \Omega^+(2d,q)$ we have $F U= \alpha(U)F$. By \Cref{thm-schur}, such $F$ is always diagonal, which proves (b), and also yields the implication $(b) \implies (c)$. If $\alpha$ is a diagonal automorphism of $G$, then clearly $\vartheta \sim \psi^{(\alpha)}$.
\end{proof}


\section{The minimal faithful permutation degree of Fitting-free groups}

\begin{definition}\label{def-mu(G,N)}
	Let $G$ be a group and let $N$ be a minimal normal subgroup of $G$ such that $N=S_1 \times \cdots \times S_{\ell}$, where the $S_i$'s are isomorphic non-abelian simple groups. Let $\phi:N_G(S_1) \rightarrow \rm{Aut}(S_1)$ be the conjugation action of  $N_G(S_1)$ on $S_1$. Let $A$ be the subgroup of $\rm{Aut}(S_1)$ induced by this conjugation action $\phi$. We define $\mu(G,N):=\mu(A).$ 
\end{definition}

We note that $S_1 \leq A \leq \rm{Aut}(S_1)$ \cite{Dixon1996}. The well-definedness of $\mu(G,N)$ follows from the fact that $S_i$'s are isomorphic under conjugation action by elements of $G$.

\begin{theorem}\label{thm: semisimple_Unique Minimal}
Let $G, K \leq \Sym(n)$ with $K \trianglelefteq G$. Let $\textbf{G} = G/K$. Suppose that $\textbf{G}$ is Fitting-free, and let $N = S_{1} \times \cdots \times S_{\ell}$ be a minimal normal subgroup of $\textbf{G}$, where the $S_{i}$'s are isomorphic non-abelian simple groups. Let $A$ be the almost simple group induced by the conjugation action of $N_{\textbf{G}}(S_1)$ on $S_1$. 

We can compute $\mu(\textbf{G},N)$ together with a minimal faithful permutation representation of $A$, in polynomial time. If furthermore, $K = 1$, then we can compute $\mu(\textbf{G}, N)$ together with a minimal faithful permutation representation of $A$, in $\textsf{NC}$.
\end{theorem}

Before we go into the proof of this theorem, we first show how the above theorem can be used to compute $\mu(G)$ of a given Fitting-free $G$, as well as a minimal faithful permutation representation $\varphi : G \to \Sym(\mu(G))$.

\begin{lemma}[{\cite[Proposition 2.5]{CannonHoltUnger}}]\label{CHU-general-Lemma} 
	Let $G$ be a Fitting-free group. Let $\mathrm{Soc}(G)=N_1 \times \cdots \times N_r$, where $N_{i}$'s are the minimal normal subgroups of $G$ with $N_{i}=S_{i1} \times \cdots \times S_{i\ell_{i}}$. Then $\mu(G)= \sum_{i=1}^{r} \ell_{i} \mu(G,N_i)$.
\end{lemma}

\begin{theorem}\label{thm:semisimple-group}
Let $G, K \leq \Sym(n)$ with $K \trianglelefteq G$. Let $\textbf{G} = G/K$. Suppose that $\textbf{G}$ is Fitting-free.  Then in polynomial time, we can compute $\mu(\textbf{G})$ and a minimal faithful permutation representation of $\textbf{G}$. Furthermore, if $K = 1$, then we can compute $\mu(\textbf{G})$ in $\textsf{NC}$, and a minimal faithful permutation representation of $\textbf{G}$ in $\textsf{RNC}$.
\end{theorem}

\begin{proof}
{\bf Determining $\mu(\textbf{G})$.} Note that $Rad(\textbf{G})=1$. We first compute $\Soc(\textbf{G})$ and its direct product decomposition into non-abelian simple groups $\Soc(\textbf{G}) = S_{1} \times \cdots \times S_{k}$. This step is polynomial-time computable in the quotients model \cite{KantorLuksQuotients}, and $\textsf{NC}$ computable when $K = 1$ (Proposition~\ref{prop:ComputeSocleFittingFree}).

Given such a decomposition, we may easily recover in $\textsf{NC}$ the minimal normal subgroups of $\textbf{G}$ by considering the conjugation action of $G$ on $S_{1}, \ldots, S_{k}$ (Proposition~\ref{prop:MinimalNormalFittingFree}). 

Let $\mathrm{Soc}(\textbf{G})=N_1 \times \cdots \times N_r$, where $N_i$'s are minimal normal subgroups of $\textbf{G}$. Suppose that $N_{i}=S_{i1} \times \cdots \times S_{i\ell_{i}}$.  From \Cref{CHU-general-Lemma} we have $\mu(\textbf{G})= \sum_{i=1}^{r} \ell_{i} \mu(\textbf{G},N_i)$. Note that given a direct product decomposition of $\Soc(\textbf{G})$ and the $N_i$'s, we can find $r$ and $\ell_{i}$'s in $\textsf{NC}$. As the $N_{i}$'s are non-abelian, we have by \Cref{thm: semisimple_Unique Minimal} that we can compute $\mu(\textbf{G},N_i)$ in polynomial-time for quotients, and in $\textsf{NC}$ if $K = 1$, as desired.

{\bf Computing a Minimal Faithful Permutation representation of $\mu(\textbf{G})$.} We construct a faithful permutation representation of $G$ on $\mu(G)=\sum_{i=1}^{r} \ell_{i} \mu(\textbf{G},N_i)$ many points as follows. Let $A_i$ be the almost simple groups induced by the conjugation action of $N_{\textbf{G}}(S_{i1})$ on $S_{i1}$. By \Cref{thm: semisimple_Unique Minimal}, we have obtained a minimal faithful permutation representation of $A_i$ on $\mu(A_i)=\mu(\textbf{G}, N_i)$ many points. Therefore, we have a generating set of $A_i$ as permutations from $\Sym(\mu(A_i))$. Let $\phi_i:N_{\textbf{G}}(S_{i1}) \to A_i$ be the homomorphism induced by the conjugation action of $N_{\textbf{G}}(S_{i1})$ on $S_{i1}$. Since $[\textbf{G}:N_{\textbf{G}}(S_{i1})]=\ell_{i}$, we can compute a set of coset representatives (i.e., a transversal $T_i$) of $N_{\textbf{G}}(S_{i1})$ in $\textbf{G}$ in polynomial time in the setting of quotients \cite{BabaiLuksSeress} (see e.g., \cite[(3.4)]{LuksReduction}), and in \textsf{RNC} in the setting of permutation groups \cite{BabaiLuksSeress}. Let $T_i'=\{N_{\textbf{G}}(S_{i1})t_i : t_i \in T_i \}$ be a set of right cosets of $N_{\textbf{G}}(S_{i1})$ in $\textbf{G}$, then $|T_i|=|T_i'|=\ell_i$. Let $\chi_i :\textbf{G} \to \Sym(\ell_{i})$ be a homomorphism induced by the natural action of $\textbf{G}$ on $T_i'$, and let $P_i$ be the subgroup of $\Sym(\ell_{i})$ induced by this action. Let $f_i: P_i \to \Sym(\ell_{i})$ be the inclusion map, sending $f_{i}(x) = x$, for all $x \in P_{i}$. We define a map $\phi_i \wr f_i : N_{\textbf{G}}(S_{i1}) \wr P_i \to A_i \wr \Sym(\ell_{i})$. 

For each $g\in G$, denote by $\bar{g}$, the unique element in $N_{\textbf{G}}(S_{i1})g \,\, \cap T_i$. Given $g$ we can compute $\overline{g}$ as follows. For each $t_i \in T_i$, we check if $t_ig^{-1} \in N_{\textbf{G}}(S_{i1})$ in $\textsf{NC}$ by Lemma~\ref{QuotientsNC}(b). We can regard $P_i$ as acting on the set $T_i$. Let $N_{\textbf{G}}(S_{i1}) \wr P_i$ be the wreath product defined using this action of $P_i$ on $T_i$. Then the base group of $N_{\textbf{G}}(S_{i1}) \wr P_i$ is ${\rm Fun}(T_i, N_{\textbf{G}}(S_{i1}))$, the set of all function from $T_i$ to $N_{\textbf{G}}(S_{i1})$. The action of $P_i$ on the base group is given by $y \to y^p$, where $y^p(t)=y(t^{p^{-1}})$ for $y \in {\rm Fun}(T_i, N_{\textbf{G}}(S_{i1}))$, $t \in T_i$, and $p \in P_i$. Then there is a monomorphism $\pi_i: \textbf{G} \to N_{\textbf{G}}(S_{i1}) \wr P_i$ defined by $\pi_i(g)=y^{\chi_i(g)}$, where $g \in \textbf{G}$ and $y \in {\rm Fun}(T_i, N_{\textbf{G}}(S_{i1}))$ defined by $y(t)=\overline {tg^{-1}} gt^{-1}$ (\cite[Section 2]{GrossKovacs1984}, see e.g., \cite[Section 2.2]{CH04}).

Define $\rho_i: \textbf{G} \to A_i \wr \Sym(\ell_{i})$ to be the composition of $\pi_i$ and $\phi_i \wr f_i$.  By \cite[Lemma 2.2]{CH04}, $\ker(\rho_i)$ is $C_{\textbf{G}}(N_i)$. Consider a map $\prod_{i=1}^{r} \rho_i: \textbf{G} \to \prod_{i=1}^{r} A_i \wr \Sym(\ell_{i})$ given by $\big(\prod_{i=1}^{r} \rho_i\big)(g)=(\rho_1(g), \ldots, \rho_r(g))$. The kernel of $\prod_{i=1}^{r} \rho_i$ is $\bigcap_{i=1}^r C_{\textbf{G}}(N_i)= C_{G}(\Soc(\textbf{G}))=1$ \cite[Section~3.1]{CH03}. Now $A_i \wr \Sym(\ell_{i})$ admits a natural embedding into $\Sym(\ell_{i} \mu(A_{i}))$ (see e.g., \cite{CH03, GrossKovacs1984}). We briefly recall the embedding here. We view the $\ell_{i}  \mu(A_i)$ many points as $\ell_{i}$ many distinct chunks, each of size $\mu(A_i)$. An element of $\Sym(\ell_{i})$ will then act on these chunks, and an element of $A_i$ will define the action within chunks (on $\mu(A_i)$ many points). 

Therefore, $\prod_{i=1}^{r} \rho_i$ is a faithful action of $G$ on $\sum_{i=1}^{r} \ell_{i} \mu(A_i)=\mu(\textbf{G})$ many points. 
\end{proof}

Now we prove~\Cref{thm: semisimple_Unique Minimal} in the remainder of this section. \\

\noindent \emph{\textbf{Proof of \Cref{thm: semisimple_Unique Minimal}}}. We follow the strategy of \cite[Theorem~4.2]{DasThakkarMPD}.
Let $N=S_1\times\cdots\times S_{\ell}$. The next task is to find the group $A$ mentioned in Definition~\ref{def-mu(G,N)}. Since $\mu(\textbf{G}, N)=\mu(A)$, our task is to find $\mu(A)$.  

As we are given the direct decomposition of $N$ into $S_1 \times \cdots \times S_{\ell}$, we have by \Cref{Normalizer} that we can, in $\textsf{NC}$ (for both the quotients and permutation groups models), compute $N_{\textbf{G}}(S_1)=\langle g_1,\ldots,g_t \rangle$. Each element $g \in N_{\textbf{G}}(S_1)$ induces an automorphism $C_g: S_1 \rightarrow S_1$ which maps an element $s$ to $g^{-1}sg$ i.e.,  $C_g(s)=g^{-1}sg$. Since $N_\textbf{G}(S_1)=\langle g_1, \ldots, g_t  \rangle$, we have  $A=\langle C_{g_{1}}, \ldots, C_{g_{t}} \rangle$ and $N_{\textbf{G}}(S_1)/C_{\textbf{G}}(S_1) \cong A$. The latter equality can be used to compute $|A|$. As $C_{N_{\textbf{G}}(S_1)}(S_1) = C_{\textbf{G}}(S_1) \trianglelefteq N_{\textbf{G}}(S_1)$, we can compute $C_{\textbf{G}}(S_1)$ in polynomial-time in the setting of quotients \cite{KantorLuksQuotients}, and in $\textsf{NC}$ for permutation groups \cite{BabaiLuksSeress}.

To compute $\mu(A)$, we use \cite[Proposition~2.2]{CannonHoltUnger}, which says that $\mu(A)$ equals $\mu(S_1)$ except for the cases mentioned in \cite[Table]{CannonHoltUnger} (cf. \Cref{MD-AlmostSimple}). If $S_1$ and $A$ do not correspond to any of the cases mentioned in the table, then $\mu(A)=\mu(S_1)$ can be computed using Theorem~\ref{thm:ConstructiveRecognition}.

Using Theorem~\ref{thm:ConstructiveRecognition}, we first determine the name of $S_1$ in polynomial time for quotients and $\textsf{NC}$ for permutation groups. Once we have the name of $S_1$, we are able to determine if $S_1$ is one of the finite simple groups in Table~\ref{Table-MD-AlmostSimple}. Next, we need to test if $A$ satisfies the condition mentioned in the 3rd column of the same row in the table. If the conditions mentioned in the \Cref{Table-MD-AlmostSimple} are satisfied, then the table itself gives the value of $\mu(A)$. Otherwise, by \Cref{MD-AlmostSimple}, we know that $\mu(A)=\mu(S_1)$, and we can compute $\mu(S_1)$ using Theorem~\ref{thm:ConstructiveRecognition}.

Note that as the groups mentioned rows 1-9 of the \Cref{Table-MD-AlmostSimple} are of constant size, these rows are easy to handle. We handle other rows separately. First, we need to check if $S_1$ is isomorphic to one of the following simple groups: $\POmegaPlus(8,q)$, $\PSL(d,q),\PSp(4,2^e), \POmegaPlus(2d,3),\G(3^e),\F(2^e),\Esix(q)$. For classical simple groups $\POmegaPlus(8,q)$, $\PSL(d,q),\PSp(4,2^e), \POmegaPlus(2d,3)$ this check is computable in $\textsf{NC}$ for both permutation groups and quotients (\Cref{thm:ConstructiveRecognition}). If $S_1$ is isomorphic to a simple group of Lie type, then its order is bounded by $n^9$. By Lemma~\ref{lem:List}, we can write down the multiplication table for $S_1$ in $\textsf{NC}$. As there are at most two finite simple groups of order $|S_1|$, we can choose prescribed permutation representations for the appropriate subset of $\G(3^e), \F(2^e), \Esix(q)$  \cite{ref4MV,ref5V,ref6V,ref7V} and apply  Lemma~\ref{lem:List} to obtain the multiplication tables for ``standard copies" of these groups. Then, using the generator enumeration strategy, we can constructively recognize the isomorphism type of $S_1$ in $\textsf{NC}$ \cite{Wolf, TangThesis} by building an isomorphism between $S_1$ and its  standard copy. Now by Theorem~\ref{thm:ConstructiveRecognition}, we can compute both $\mu(\G(3^e)), \mu(\F(2^e)), \mu(\Esix(q))$ as well as a minimal faithful permutation representation on the cosets of a maximal subgroup in $\textsf{NC}$. Therefore, we can check in $\textsf{NC}$ if $S_1$ isomorphic to one of $\G(3^e)$, $\F(2^e)$, $\Esix(q)$.

So we can now assume that we have identified if $S_1$ is isomorphic to one of the simple groups mentioned in the 2nd column of Table~\ref{Table-MD-AlmostSimple} in $\textsf{NC}$. We now need to check if $A$ satisfies the corresponding conditions in the 3rd column of Table~\ref{Table-MD-AlmostSimple}. We handle this check for each of the cases in the next few sections: $\POmegaPlus(8,q)$ (Section~\ref{sec:POmegaPlus}); $\PSL(d,q)$ (Section~\ref{PSL-case}); $\PSp(4, 2^e)$ (Section~\ref{PSp-case}); $\POmegaPlus(2d,3)$ (Section~\ref{POmega-case}); and $\G, \F, \Esix$ (Section~\ref{sec:exceptional-groups}). 

While computing the minimal faithful permutation degree of the non-abelian simple groups not listed on Table~\ref{Table-MD-AlmostSimple} poses no challenge, we must still build minimal faithful permutation representations for the corresponding almost simple groups. We handle this in Section~\ref{sec:RemainingSimpleGroups}.

\section{\texorpdfstring{$S_1 \cong \POmegaPlus(8,q)$}.} \label{sec:POmegaPlus}
 
In this section, we will establish the following.

\begin{proposition} \label{prop:5.1}
Let $G, K \leq \text{Sym}(n)$, with $K \trianglelefteq G$. Let $\textbf{G} := G/K$ be Fitting-free, with minimal normal subgroup $N = S_{1} \times \cdots \times S_{\ell}$, where the $S_{i}$'s are isomorphic non-abelian simple groups, isomorphic to $\POmegaPlus(8,q)$ for some prime power $q \geq 4$. Furthermore, suppose that we are given the following:
\begin{itemize}
\item Generators for each of the $S_{i}$'s.
\item A minimal faithful permutation representation $\text{Iso} : S_{1} \to S_{1}$, detailing the action of $S_{1}$ on the set $\overline{V}$ of $1$-dimensional anisotropic subspaces of $V = \mathbb{F}_{q}^{8}$. 

\item A basis $\beta$ for $V$, as well as the underlying quadratic form $Q$ and the associated bilinear form $\mathfrak{f}_{Q}$.

\item $N_{\textbf{G}}(S_1) = \langle g_1, \ldots, g_t \rangle$, as well as $C_{\textbf{G}}(S_1)$.
\item For each $i \in [t]$, let $C_{g_{i}} : S_1 \to S_1$ be given by $C_{g_i}(s) = g_{i}^{-1}sg_{i}$. Suppose that we are given $A = \langle C_{g_{i}}, \ldots, C_{g_{t}} \rangle$.
\end{itemize}

\noindent Then we can compute $\mu(\textbf{G},N) := \mu(A)$, as well as a minimal faithful permutation representation of $A$, in $\textsf{NC}$. 
\end{proposition}

We first recall some preliminaries from \cite{HallTrialityNotes}. An \emph{algebra} $\mathfrak{A}$ over the field $\mathbb{F}$ is an $\mathbb{F}$-vector space, together with a bilinear product $\pi : \mathfrak{A} \times \mathfrak{A} \to \mathfrak{A}$. The algebra admits \emph{composition} if there is a nondegenerate quadratic form $Q : \mathfrak{A} \to \mathbb{F}$ such that for all $a, b \in \mathfrak{A}$, $Q(a)Q(b) = Q(ab)$. In particular, $Q(1) = 1$.

We now turn to proving Proposition~\ref{prop:5.1}.
\begin{proof}[Proof of Proposition~\ref{prop:5.1}]
\noindent 
\begin{itemize}
\item \textbf{Determining $\mu(A)$.} From the 3rd column, we need to check if $3$ divides $|A /S_1|$. We can determine $|S_1|$, $|N_G(S_1)|$, $|C_G(S_1)|$ and $|A /S|$ in $\textsf{NC}$ (Lemma~\ref{QuotientsNC}). If $3$ divides $|A/S|$ then then $\mu(G,N)= 3\mu(S_1)$ otherwise $\mu(G,N)= \mu(S_1)$.

\item \textbf{Computing a minimal faithful permutation representation of $A$.} We now turn to computing a minimal faithful permutation representation of $A$. For each $i \in [t]$, define $\lambda_{g_{i}} : \text{Iso} \circ C_{g_{i}} \circ \text{Iso}$ ($i \in [t]$). Let $U \in \mathrm{\Omega^{+}}(8,q)$ with rows and columns of $U$ are numbered by $1,\ldots, d,$ $ -1, \ldots, -d$. Let:
\begin{align*}
L=\{&I+\beta(e_{i,j}-e_{-j,-i}), \\
&I-\beta(e_{-i,-j}-e_{j,i}), \\
&I+\beta(e_{i,-j}-e_{j,-i}), \\
&I+\beta(e_{-i,j}-e_{-j,i})  \mid 0 < i < j, \beta \in \mathbb{F}_3\},    
\end{align*} 
where $e_{i,j}$ is the matrix which has $1$ in the $(i, j)$th entry and $0$ in other places. Then $\mathrm{\Omega^{+}}(8,q)=\langle L \rangle$ (see e.g., \cite[Page~185]{carter-book}).

For each $i \in [t]$, we apply Lemma~\ref{construct-alpha-POmega} on $\lambda_{g_{i}}$ to obtain, in $\textsf{NC}$, $\alpha_{g_{i}} \in \Aut(\Omega^{+}(8,q))$ such that $\lambda_{g_{i}} = \overline{\alpha_{g_{i}}}$. Each $\alpha_{g_{i}}$ is specified by its action on $L$. The subgroup $\langle \alpha_{g_1}, \ldots, \alpha_{g_t} \rangle$ is the natural embedding of $A$ into $\Aut(\Omega^{+}(8,q))$. We will construct permutation representations for each $\alpha_{g_{i}}$ ($i \in [t]$).

From here on out, we will use $\mathcal{A} := \mathcal{A}_{V}$ to refer to the algebra arising from $(V, Q)$. For $x \in \mathfrak{A}$, the \emph{conjugate} of $x$ is $\overline{x} := -x + \mathfrak{f}_{Q}(x,1)1$. Let $\mathcal{S}$ be the set of nonzero singular vectors with respect to $\mathfrak{f}_{Q}$. Note that $\mathcal{S}$ is $\textsf{NC}$-computable \cite{KantorLuksMark,MarkThesis}.

We now turn to specifying the \emph{triality automorphisms} $\tau$ and $\kappa$ (see e.g., \cite{HallTrialityNotes}). Define $M^{\lambda} := \{ x\mathfrak{A} : x \in \mathcal{S} \}$ and $M^{\rho} := \{ \mathfrak{A}x : x \in \mathcal{S} \}$, and let $\mathcal{T} := V^{*} \dot\cup M^{\lambda} \dot\cup M^{\rho}$. Now $\tau$ acts on $\mathcal{T}$ as follows:
\[
[x] \xrightarrow{\tau} \overline{x}\mathfrak{A} \xrightarrow{\tau}\mathfrak{A}\overline{x} \xrightarrow{\tau} [x],
\]
where we use $[x]$ to denote the $1$-space generated by $x \in V$. Furthermore, define the permutation $\kappa$ on $\mathcal{T}$ as follows: 
\[
\kappa([x]) = [\overline{x}]; \kappa(x \mathfrak{A}) = \mathfrak{A} \overline{x};  \kappa(\mathfrak{A} x) = \overline{x} \mathfrak{A}.
\]
Note that $\kappa \tau \kappa = \tau^{-1}$ \cite[Proposition~5.19]{HallTrialityNotes}. Additionally, notice that we have explicit permutation representations of $\tau, \kappa$ on $\mathcal{T}$.

Let $q = p^{e}$, where $p$ is a prime. Let us consider the Frobenius automorphism $\sigma_{0}$ of $\Aut(\mathbb{F}_{q})$ given by $a \mapsto a^p$. The automorphism $\sigma_{0}$ is a generator of $\Aut(\mathbb{F}_{q})$ of order $e$. For each $\alpha_{g_{i}}$, define $\alpha_{g_{i}}^{t', t'', t'''} = \alpha_{g_{i}} \mathfrak{f}_{\sigma_{0}}^{-t'} \tau^{-t''} \kappa^{-t'''}$, where $\mathfrak{f}_{\sigma_{0}}$ is a field automorphism; $\tau$ and $\kappa$ are as defined above; and $0 \leq t' < e, 0 \leq t'' \leq 2, t''' \in \{0,1\}$. The maps $\alpha_{g_{i}}^{t', t'', t'''}$ help us to identify the type of automorphism of $\alpha_{g_{i}}$. 

Now for each $t', t'', t'''$, the automorphism $\alpha_{g_{i}}^{t', t'', t'''}$ is specified on a generating set $L$. Observe that there is exactly one choice of $t', t'', t'''$ for which $\alpha_{g_{i}}^{t', t'', t'''}$ is a diagonal automorphism. For each $t', t'', t'''$, we use Observation~\ref{Observation-typechecking-POmega} with Lemma~\ref{lemma-module-equ-cond-POmega} (with $B = L$) to check if $\alpha_{g_{i}}^{t', t'', t'''}$ is a diagonal automorphism. This check can be handled in $\textsf{NC}^{2}$. Let $t', t'', t'''$ such that $\alpha_{g_{i}}^{t', t'', t'''}$ is an inner or diagonal automorphism. Note that if $t'' = t''' = 0$, then $\alpha_{g_{i}} = \alpha_{g_{i}}^{t', t'', t'''} f_{\sigma_{0}}^{t'}$ does not have $\kappa$ or $\tau$ (in which case, we say that $\alpha_{g_{i}}$ is \emph{free of the triality automorphisms}). Note that $h_{i} := \alpha_{g_{i}}^{t', t'', t'''} f_{\sigma_{0}}^{t'} \in \mathrm{P}\Gamma \mathrm{L}(d,q)$, and so is (abstractly) a permutation of $V$ (though at this stage, need not be represented as a permutation of $V$). As we have fixed a basis $\beta$ of $V$, we may apply \cite[Lemma~6.5(iii)]{KantorLuksMark} to obtain $M = (a_{ij}) \leq \text{GL}(V)$ and $\sigma \in \Aut(\mathbb{F}_{q})$ such that the pair $(M, \sigma)$ induces the semilinear transformation $\gamma : \sum_{i=1}^{8} c_{i}b_{i} \mapsto \sum_{i=1}^{8} c_{i}^{\sigma} a_{ij}b_{j}$. 

We will now construct a permutation representation $\psi_{i}$ of $h_{i}$ as follows. We have that $h_{i}$ acts on the anisotropic $1$-spaces $[v]$ of $V$ by sending $[v] \mapsto [\gamma v]$. Furthermore, $h_{i}$ acts on $M^{\lambda}$ by sending $(x\mathfrak{A}) \mapsto (\gamma x)\mathfrak{A}$; and $h_{i}$ acts on $M^{\rho}$ by sending $\mathfrak{A}x \mapsto \mathfrak{A}( (\gamma^{-1})^{T} x)$.

Now if $(t'', t''') \neq (0, 0)$, we multiply $\psi_{i}$ with the permutation representation of $\tau^{t''}\kappa^{t'''}$.

If $\mu(A) = \mu(S)$, then we restrict the permutations $\psi_{1}, \ldots, \psi_{t}$ to the points of $V$. Otherwise, we have constructed a minimal faithful permutation representation of $A$ on $3\mu(S)$ points. \qedhere

\end{itemize}
\end{proof}

\section{\texorpdfstring{$S_1 \cong \PSL(d,q), d \geq 3, (d,q) \neq (3,2), (4,2)$}.}\label{PSL-case}

\begin{proposition} \label{prop:5.2}
Let $d \geq 3$, and $q$ be a prime power, such that $(d,q) \neq (3,2), (4, 2)$. Let $G, K \leq \text{Sym}(n)$, with $K \trianglelefteq G$. Let $\textbf{G} := G/K$ be Fitting-free, with minimal normal subgroup $N = S_{1} \times \cdots \times S_{\ell}$, where the $S_{i}$'s are isomorphic non-abelian simple groups, isomorphic to $\PSL(d,q)$ for some prime power $q$. Furthermore, suppose that we are given the following:
\begin{itemize}
\item Generators for each of the $S_{i}$'s.
\item A minimal faithful permutation representation $\text{Iso} : S_{1} \to S_{1}$, detailing the action of $S_{1}$ on the set $\overline{V}$ of $1$-dimensional subspaces of $V = \mathbb{F}_{q}^{d}$. In particular, we are given $V$ explicitly. 

\item A basis $\beta$ for $V$.

\item $N_{\textbf{G}}(S_1) = \langle g_1, \ldots, g_t \rangle$, as well as $C_{\textbf{G}}(S_1)$.
\item For each $i \in [t]$, let $C_{g_{i}} : S_1 \to S_1$ be given by $C_{g_i}(s) = g_{i}^{-1}sg_{i}$. Suppose that we are given $A = \langle C_{g_{i}}, \ldots, C_{g_{t}} \rangle$.
\end{itemize}

\noindent Then we can compute $\mu(\textbf{G},N) := \mu(A)$, as well as a minimal faithful permutation representation of $A$, in $\textsf{NC}$.
\end{proposition}

\begin{proof}
\noindent
\begin{itemize}
\item \textbf{Determining $\mu(A)$.} To begin, we will outline our approach in this case. Our goal is to determine whether \( A \not \leq \PGammaL \), as indicated in the third column. It is known that \( \PGammaL \leq \aut(\PSL(d,q)) \cong \aut(\mathrm{SL}(d,q)) \). Given that \( A \leq \aut(S_1) \) and \( S_1 \cong \PSL(d,q) \), we can represent \( A \) as a subgroup of \( \aut(\PSL(d,q)) \) via the isomorphism \( {\rm{Iso}} \). Furthermore, we have that \( \aut(\PSL(d,q)) \cong \aut(\mathrm{SL}(d,q)) \), allowing us to also view \( A \) as a subgroup of \( \aut(\mathrm{SL}(d,q)) \).

Now, we have both \( A \) and \( \PGammaL \) as subgroups of \( \aut(\mathrm{SL}(d,q)) \). Furthermore, the automorphisms that contain the graph automorphism $\mathbbm{g}$ in their decompositions are exactly the elements in $\aut(\mathrm{SL}(d,q))\setminus \PGammaL$ (see \Cref{subsec: simplegroup}). Therefore, to verify whether \( A \) is a subgroup of \( \PGammaL \), it suffices to check that each generator of \( A \) is free from the graph automorphism.

Let $\mathrm{SL}(d,q)=\langle L \rangle$, where $L=\{I+\beta e_{i,j} \mid \beta \in \mathbb{F}_{q}^{*}, i,j \in [d], i \neq j\}$, $e_{i,j}$ is the matrix which has $1$ in the $(i,j)$th entry and $0$ in the remaining entries (see e.g., \cite[Page~185]{carter-book}). Define $\lambda_{g_{i}} \in \aut(\PSL(d,q))$ as $\lambda_{g_{i}}= \mathrm{Iso} \circ C_{g_{i}} \circ \mathrm{Iso}^{-1}$. Thus, $A$ is embedded as the subgroup $\langle \lambda_{g_{1}}, \ldots, \lambda_{g_{t}} \rangle$ in  $\aut(\PSL(d,q))$. Now we apply Lemma~\ref{construct-alpha} on each $\lambda_{g_i}$ to get $\alpha_{g_i} \in \aut(\mathrm{SL}(d,q))$  such that $\lambda_{g_i}=\bar{\alpha}_{g_i}$. This step is computable in polynomial-time in the setting of quotients and $\textsf{NC}$ in the setting of permutation groups. Each $\alpha_{g_i}$ is specified by its action on $L$. The subgroup $\langle \alpha_{g_1},\ldots, \alpha_{g_t}\rangle$ is the natural embedding of $A$ in $
\aut(\mathrm{SL}(d,q))$. 

Let $q=p^e$. Let us consider the Frobenius automorphism $\sigma_0$ of $\aut(\mathbb{F}_q)$ given by $a \mapsto a^p$. The automorphism $\sigma_0$ is a generator of $\aut(\mathbb{F}_q)$ of order $e$. For each $\alpha_{g_{i}}$, define $\alpha_{g_{i}}^{t',t''}=\alpha_{g_{i}} \mathbbm{f}_{\sigma_0}^{-t'}\mathbbm{g}^{-t''}$, where $\mathbbm{f}_{\sigma_0}$ is a field automorphism, $\mathbbm{g}$ is a graph automorphism, and $0 \leq t' < e$, $t'' \in \{0,1\}$. The maps $\alpha_{g_{i}}^{t',t''}$ help us to identify the type of the automorphism $\alpha_{g_{i}}$. 

For each $i$, and each $t',t''$, the automorphism $\alpha_{g_{i}}^{t',t''}$ is specified on a generating set $L$. Observe that for a fixed $i$, there is exactly one choice of $t'$ and $t''$ such that $\alpha_{g_{i}}^{t',t''}$ is a diagonal automorphism or an inner automorphism. For each $t',t''$, we use \Cref{Observation-typechecking-PSL} and \Cref{lem-type-checking} (with $B=L$) to check if $\alpha_{g_{i}}^{t',t''}$ is an inner or a diagonal automorphism. This check can be handled in $\textsf{NC}^2$. Let $t',t''$ be such that $\alpha_{g_{i}}^{t',t''}$ is either an inner automorphism or a diagonal automorphism. Note that $t''=0$ if and only if $\alpha_{g_{i}}=\alpha_{g_{i}}^{t',t''} \circ  \mathbbm{f}_{\sigma_0}^{t'}$ and it does not have a graph automorphism in its decomposition, i.e., $\alpha_{g_{i}}  \in \PGammaL$.  

If for each $i$, $\alpha_{g_i}$ turns out to be in $\PGammaL$ we can conclude that $A\leq \PGammaL$. Thus, we can test if $A \not \leq \PGammaL$ in polynomial-time in the setting of quotients and $\textsf{NC}$ in the setting of permutation groups. If $A \leq \PGammaL$ then $\mu(\textbf{G},N)=\mu(S_1)$. Otherwise, $\mu(\textbf{G},N)=2\mu(S_1)$ (see e.g., \cite{CannonHoltUnger}).

\item \noindent \textbf{Computing a Minimal Faithful Permutation Representation of $A$.} We now turn computing a minimal faithful permutation representation of $A$. We first turn to specifying the permutation representation of the graph automorphism, $\iota_{\beta} : \PSL(d,q) \to \PSL(d,q)$ sending $M \mapsto (M^{-1})^{T}$. Recall that $\mu(A) = 2\mu(S)$ if and only if some generator of $A$ contains the graph automorphism. In particular, we require $2\mu(S)$ points to construct a minimal faithful permutation representation of $\iota_{\beta}$. We accomplish this in the following manner. Let $V^{*} = \text{Hom}(V, \mathbb{F}_{q})$ be the dual space of $V$. As we are given $V$ explicitly, we may construct $V^{*}$ in $\textsf{NC}$. For a subspace $U \leq V$, define:
\[
U^{\perp} = \{ \ell : V \to \mathbb{F}_{q} \mid \ell(u) = 0 \text{ for all } u \in U\}.
\]
For each $1$-space $[v] \leq V$, we consider $([v], +)$ to be the corresponding $1$-space in $V$, and $([v], -)$ to correspond to $[v]^{\perp}$. Now $\iota_{\beta}$ swaps $([v], +)$ and $([v], -)$ for all $1$-spaces $[v]$ of $V$. 

While computing $\mu(A)$ we have computed a decomposition of $\alpha_{g_{i}}$ as $\alpha_{g_{i}}^{t',t''}\mathbbm{g}^{t''} \mathbbm{f}_{\sigma_0}^{t'}$. However, we can also compute the decomposition of $\alpha_{g_{i}}$ as $\alpha_{g_{i}}^{t',t''} \mathbbm{f}_{\sigma_0}^{t'} \mathbbm{g}^{t''}$ as in \Cref{baarnhielm Thm-decomposition} similar to the above. With this later decomposition, suppose for each $i \in [t]$, let $t', t''$ such that $\alpha_{g_{i}}^{t', t''}$ is inner or diagonal. In particular, $\alpha_{g_{i}}^{t', t''}\mathfrak{f}_{\sigma_{0}}^{t'} \in \PGammaL$, and so $\alpha_{g_{i}}^{t', t''}\mathfrak{f}_{\sigma_{0}}^{t'}$ is (abstractly) a permutation of $V$ (though at this stage, need not be represented as a permutation of $V$). As we have fixed a basis $\beta$ of $V$, we may apply \cite[Lemma~6.5(iii)]{KantorLuksMark} to obtain $M = (a_{ij}) \leq \GL(V)$ and $\sigma \in \Aut(\mathbb{F}_{q})$ such that the pair $(M,\sigma)$ induces the semilinear transformation $\tau : \sum_{i=1}^{d} c_{i}e_{i} \mapsto \sum_{i=1}^{d} c_{i}^{\sigma} a_{ij} e_{j}$. By applying $\tau$ to $V$, we may construct a permutation representation $\psi_{i}$ of $\alpha_{g_{i}}^{t', t''}\mathfrak{f}_{\sigma_{0}}^{t'}$. Precisely, $\alpha_{g_{i}}^{t', t''}\mathfrak{f}_{\sigma_{0}}^{t'}$ maps points of the form $([v], +) \mapsto ([\tau v], +)$, and points of the form $([v], -) \mapsto ([(\tau^{-1})^{T}v], -)$.

Now if $t'' \neq 0$, then $\alpha_{g_{i}}$ has a graph automorphism. So we multiply $\psi_{i}$ with the permutation representation of $\iota_{\beta}$.

If $\mu(A) = \mu(S)$, then we restrict the permutation representations $\psi_{1}, \ldots, \psi_{t}$ to the points of the form $([v], +)$. Otherwise, we have constructed a minimal faithful permutation representation of $A$ on $2\mu(S)$ points. \qedhere
\end{itemize}
\end{proof}

\section{\texorpdfstring{$S_1 \cong \PSp(4,2^e), e \geq 2$}.}\label{PSp-case}

\begin{proposition} \label{prop:PSp}
Let $e \geq 2$. Let $G, K \leq \text{Sym}(n)$ with $K \trianglelefteq G$. Let $\textbf{G} := G/K$ be Fitting-free, with minimal normal subgroup $N = S_{1} \times \cdots \times S_{\ell}$, where the $S_{i}$'s are all isomorphic non-abelian simple groups, isomorphic to $\PSp(4,2^e)$. Furthermore, suppose that we are given the following.
\begin{itemize}
\item Generators for each of the $S_i$'s.
\item A minimal faithful permutation representation $\text{Iso} : S_1 \to S_1$, detailing the action of $S_1$ on the $1$-spaces $\overline{V}$ of a symplectic space $V = \mathbb{F}_{2^{e}}^{4}$. In particular, we are given $V$ explicitly by its full set of elements.

\item $N_{\textbf{G}}(S_1) = \langle g_1, \ldots, g_t \rangle$, as well as $C_{\textbf{G}}(S_1)$.

\item For each $i \in [t]$, let $C_{g_{i}} : S_1 \to S_1$ be given by $C_{g_{i}}(s) = g_{i}^{-1}sg_{i}$. Suppose that we are given $A = \langle C_{g_{1}}, \ldots, C_{g_{t}} \rangle$.
\end{itemize}

Then we can compute $\mu(\textbf{G}, N) := \mu(A)$, as well as a minimal faithful permutation representation of $A$, in $\textsf{NC}$.
\end{proposition}

The case when $S_1 \cong \PSp(4,2^e)$ is handled similarly to the case of $S_1 \cong \PSL(d,q)$. For completeness, we provide the details below.

\begin{proof}
\noindent 
\begin{itemize}
    
\item \textbf{Determining $\mu(A)$.} To begin, we will outline our approach in this case. Note that  $Z=Z(\rm{Sp}(4,2^e))=1$. Thus, $\PSp(4,2^e) =\mathrm{Sp}(4,2^e)$, and so $\aut(\PSp(4,2^e))=\aut(\rm{Sp}(4,2^e))$.  From the third column,  our aim is to decide if $A \not \leq \rm{P}\Gamma\rm{Sp}(4,2^e)$. Since $S_1 \cong \PSp(4,2^e)$, we can view $A$ as a subgroup of $\aut(\PSp(4,2^e))= \aut(\rm{Sp}(4,2^e))$. Thus, we have $A, \rm{P}\Gamma\rm{Sp}(4,2^e) \leq \aut(\rm{Sp}(4,2^e))$ \cite{KleidmanLiebeck,carter-book}.

As described in \Cref{def:automorphisms}, there are four types of automorphisms. Also, any automorphism of $\rm{Sp}(4,2^e)$ can be decomposed as $ \mathbbm{i}\mathbbm{d}\mathbbm{f}\mathbbm{g}$ \cite{carter-book}. The automorphisms of $\rm{Sp}(4,2^e)$ containing the graph automorphism $\mathbbm{g}$ in its decomposition are exactly the elements in $\aut(\rm{Sp}(4,2^e))\setminus \rm{P}\Gamma\rm{Sp}(4,2^e)$ \cite{KleidmanLiebeck,carter-book}. Therefore, to check if $A$ is a subgroup of $\rm{P}\Gamma\rm{Sp}(4,2^e)$, it is enough to check if each of the generators of $A$ is free of the graph automorphism. 

We now provide the computational details. For a matrix $U \in \rm{Sp}(4,2^e)$, numbering the row and columns by $1,\ldots, d, -1, \ldots, -d$. Let
 \begin{align*}
 L=\{&I+\beta(e_{i,j}-e_{-j,-i}),\\ 
 &I-\beta(e_{-i,-j}-e_{j,i}),\\ 
 &I+\beta(e_{i,-j}+e_{j,-i}),\\
 &I-\beta(-e_{-i,j}-e_{-j,i}),\\ 
 &I+\beta e_{i,-i},\\ 
 &I+\beta e_{-i,i}  \mid 0 < i < j, \beta \in \mathbb{F}_{2^e}\},
 \end{align*}
 where $e_{i,j}$ is the matrix which has $1$ in the $(i, j)$th entry and $0$ in other places. Then $\rm{Sp}(4,2^e)=\langle L \rangle$ (see e.g., \cite[Page~186]{carter-book}). Note that for each $U \in L$, $\ord(U)=2$.

Define $\lambda_{g_{i}} \in \aut(\PSp(4,2^e))=\aut(\rm{Sp}(4,2^e))$ as $\lambda_{g_{i}}= \mathrm{Iso} \circ C_{g_{i}} \circ \mathrm{Iso}^{-1}$. Thus $A \cong \langle \lambda_{g_{1}}, \ldots, \lambda_{g_{t}} \rangle \leq \aut(\rm{Sp}(4,2^e))$. For each $\lambda_{g_{i}}$, define $\alpha_{g_{i}}^{t',t''}=\lambda_{g_{i}} \mathbbm{f}_{\sigma_0}^{-t'}\mathbbm{g}^{-t''}$ such that $0 \leq t' < e$ and $t'' \in \{0,1\}$. Our aim is to identify the type of $\lambda_{g_{i}}$. Note that it is enough to identify the type of $\alpha_{g_{i}}^{t',t''}$.

For each $i$, and each $t',t''$, the automorphism $\alpha_{g_{i}}^{t',t''}$ is specified on a generating set $L$. Observe that for a fixed $i$, there is exactly one choice of $t'$ and $t''$ such that $\alpha_{g_{i}}^{t',t''}$ is a diagonal automorphism or an inner automorphism. For each $t',t''$, we use \Cref{Observation-typechecking-PSL} and \Cref{lem-type-checking} (with $B=L$) to check if $\alpha_{g_{i}}^{t',t''}$ is a diagonal automorphism. This check is computable in $\textsf{NC}$. Let $t',t''$ be such that $\alpha_{g_{i}}^{t',t''}$ is an inner or a diagonal automorphism. One can see that $t''=0$ if and only if $\lambda_{g_{i}}=\alpha_{g_{i}}^{t',t''} \mathbbm{f}_{\sigma_0}^{t'}$, and $\lambda_{g_{i}}$ is free of graph automorphism i.e., $\lambda_{g_{i}} \in \rm{P}\Gamma\rm{Sp}(4,2^e)$.  

We repeat the above procedure for each $i$ to determine the type of $\lambda_{g_{i}}$. Suppose for each $i$, we get $\alpha_{g_{i}}^{t',t''}$ is a diagonal automorphism or an inner automorphism, and $t''=0$. Then $\lambda_{g_{i}} \in \rm{P}\Gamma\rm{Sp}(4,2^e)$. Therefore, $A \leq \rm{P}\Gamma\rm{Sp}(4,2^e)$. Thus we can test if $A \not \leq \rm{P}\Gamma\rm{Sp}(4,2^e)$ in $\textsf{NC}$. If $A \leq \rm{P}\Gamma\rm{Sp}(4,2^e)$ then $\mu(G,N)=\mu(S_1)$, otherwise, $\mu(G,N)=2\mu(S_1)$ (see e.g., \cite{CannonHoltUnger}).

\item \textbf{Computing a Minimal Faithful Permutation Representation of $A$.} We now turn  computing a minimal faithful permutation representation of $A$. As $V = \mathbb{F}_{2^e}^{4}$ is given explicitly, we can in $\textsf{NC}$ compute a basis $\{e_1,e_2,f_1,f_2\}$ such that $\langle e_i, e_j \rangle=\langle f_i, f_j \rangle = 0$ and $\langle e_i, f_j \rangle =\delta_{ij}$, where $\delta_{ij}=1$ if and only if $i=j$; otherwise $\delta_{ij} = 0$. We now turn to specifying the permutation representation of the graph automorphism, $\iota_{\beta} : \PSp(4,2^e) \to \PSp(4,2^e)$ sending $M \mapsto (M^{-1})^{T}$. 

Recall that $\mu(A) = 2\mu(S)$ if and only if some generator of $A$ contains the graph automorphism. In particular, we require $2\mu(S)$ points to construct a minimal faithful permutation representation of $\iota_{\beta}$. We accomplish this in the following manner. For a subspace $U \leq V$, define:
\[
U^{\perp} = \{ v \in V : \langle v,u \rangle =0  \text{ for all } u \in U\}.
\]
For each one-space $[v]$ of $V$, we construct elements $([v], +), ([v], -)$. Now $\iota_{\beta}$ swaps $([v], +)$ and $([v], -)$ for all $1$-spaces of $V$ (see e.g., \cite[Section 7.1]{cameronnotes}).

While computing $\mu(A)$ we have computed a decomposition of $\alpha_{g_{i}}$ as $\alpha_{g_{i}}^{t',t''}\mathbbm{g}^{t''} \mathbbm{f}_{\sigma_0}^{t'}$. However, we can also compute the decomposition of $\alpha_{g_{i}}$ as $\alpha_{g_{i}}^{t',t''} \mathbbm{f}_{\sigma_0}^{t'} \mathbbm{g}^{t''}$ as in \Cref{baarnhielm Thm-decomposition} similar to the above. With this later decomposition, suppose for each $i \in [t]$, let $t', t''$ such that $\alpha_{g_{i}}^{t', t''}$ is inner or diagonal. In particular, $\alpha_{g_{i}}^{t', t''}\mathfrak{f}_{\sigma_{0}}^{t'} \in \rm{P}\Gamma\rm{Sp}(4,2^e)$, and so $\alpha_{g_{i}}^{t', t''}\mathfrak{f}_{\sigma_{0}}^{t'}$ is (abstractly) a permutation of $V$ (though not necessarily represented as such at this stage). As we have fixed a basis $\beta$ of $V$, we may apply \cite[Lemma~6.5(iii)]{KantorLuksMark} to obtain $M = (a_{ij}) \leq \GL(V)$ and $\sigma \in \Aut(\mathbb{F}_{q})$ such that the pair $(M,\sigma)$ induces the semilinear transformation $\tau : \sum_{i=1}^{d} c_{i}e_{i} \mapsto \sum_{i=1}^{d} c_{i}^{\sigma} a_{ij} e_{j}$. By applying $\tau$ to $V$, we may construct a permutation representation $\psi_{i}$ of $\alpha_{g_{i}}^{t', t''}\mathfrak{f}_{\sigma_{0}}^{t'}$. Precisely, $\alpha_{g_{i}}^{t', t''}\mathfrak{f}_{\sigma_{0}}^{t'}$ maps points of the form $([v], +) \mapsto ([\tau v], +)$, and points of the form $([v], -) \mapsto ([(\tau^{-1})^{T}v], -)$.

Now if $t'' \neq 0$, then $\alpha_{g_{i}}$ has a graph automorphism. So we multiply $\psi_{i}$ with the permutation representation of $\iota_{\beta}$ 

If $\mu(A) = \mu(S)$, then we restrict the permutation representations $\psi_{1}, \ldots, \psi_{t}$ to the points of the form $([v], +)$. Otherwise, we have constructed a minimal faithful permutation representation of $A$ on $2\mu(S)$ points. \qedhere
\end{itemize}

\end{proof}

\section{\texorpdfstring{$S_1 \cong  \POmegaPlus(2d,3), d \geq 4$}.}\label{POmega-case}

\begin{proposition} \label{prop:POmega+2d3}
Let $d \geq 4$ and let $G, K \leq \text{Sym}(n)$, with $K \trianglelefteq G$. Let $\textbf{G} := G/K$ be Fitting-free, with minimal normal subgroup $N = S_{1} \times \cdots \times S_{\ell}$, where the $S_{i}$'s are isomorphic non-abelian simple groups, isomorphic to $\POmegaPlus(2d,3)$. Furthermore, suppose that we are given the following:
\begin{itemize}
\item Generators for each of the $S_{i}$'s.
\item A minimal faithful permutation representation $\text{Iso} : S_{1} \to S_{1}$, detailing the action of $S_{1}$ on the set $\overline{V}$ of $1$-dimensional anisotropic subspaces of $V=\mathbb{F}_{3}^{2d}$ .

\item A basis $\beta$ for $V$.

\item $N_{\textbf{G}}(S_1) = \langle g_1, \ldots, g_t \rangle$, as well as $C_{\textbf{G}}(S_1)$.
\item For each $i \in [t]$, let $C_{g_{i}} : S_1 \to S_1$ be given by $C_{g_i}(s) = g_{i}^{-1}sg_{i}$. Suppose that we are given $A = \langle C_{g_{i}}, \ldots, C_{g_{t}} \rangle$.
\end{itemize}

\noindent Then we can compute $\mu(\textbf{G},N) := \mu(A)$, as well as a minimal faithful permutation representation of $A$, in $\textsf{NC}$.
\end{proposition}

\begin{proof}
\noindent 
\begin{itemize}
    
\item \textbf{Determining $\mu(A)$.} Our aim in this case is to determine if $3 \nmid |A|/|S|$ and  $A \not \leq {\rm{PO}}^+(2d,3)$ in the 3rd column of Table~\ref{Table-MD-AlmostSimple}. Note that as $A\leq \aut(S_1)$ and $S_1\cong \POmegaPlus(2d,3)$, using ${\rm{Iso}}$, we can view $A \leq \aut(\POmegaPlus(2d,3))$. By  \Cref{Thm-Iso-POmegga-Omega}, we have $A$ as a subgroup of $\aut(\Omega^+(2d,3))$. Note that ${\rm{PO}}^{+}(2d,3) \leq \Aut(\Omega^{+}(2d,3))$. 

Each element of $\aut(\Omega^{+}(2d,3))$ can be decomposed in terms of inner, diagonal, field, and graph automorphisms (cf. \Cref{def:automorphisms}). However, in the case of $\Omega^+(2d,3)$, the field has size $3$. Thus, it does not admit any non-trivial field automorphisms. Also, $\Omega^+(2d,3)$ does not have any graph automorphisms \cite{KleidmanLiebeck}. 

Suppose $X = \begin{pmatrix} 0_{d} & I_{d}\\ I_{d} & 0_{d} \end{pmatrix}$  is the matrix of the symmetric bilinear form fixed by the element of $\rm{O}^+(2d,3)$, then
$\rm{O}^+(2d,3) =\{F \in M_{2d}(3) : FXF^{t}=X \}$ \cite{KleidmanLiebeck,dieudonne1951}. Let $Z=Z({\rm{O}}^+(2d,3)) \cap \Omega^+(2d,3)$. Note that any automorphism of $\Omega^+(2d,3)$ is induced by the projective image of the matrices $F \in \mathrm{GL}(2d,3)$ such that $FX F^{t} = \pm X$. Here, we use the fact that ${\rm PO}^{+}(2d,3)={\rm O}^{+}(2d,3)/Z({\rm O}^{+}(2d,3))$ can be viewed as a subgroup of $\aut(\Omega^{+}(2d,3))$, and it contains those automorphisms of $\aut(\Omega^{+}(2d,3))$ that fix the symmetric bilinear form. Namely, let $FZ \in  {\rm PO}^{+}(2d,3)$, and let $\alpha_{F}$ be an induced automorphism obtained from $F$. If $FX F^{t} = X$, then $\alpha_{F} \in {\rm PO}^{+}(2d,3)$. Therefore, to test if $A\leq {\rm{PO}}^+(2d,3)$ it is enough to check if the generators of $A$ are induced by the projective image of $F \in {\rm O}^{+}(2d,3)$.

We now describe the computational details. Let $U \in \mathrm{\Omega^{+}}(2d,3)$ with rows and columns of $U$ are numbered by $1,\ldots, d,$ $ -1, \ldots, -d$. Let:
\begin{align*}
L=\{&I+\beta(e_{i,j}-e_{-j,-i}), \\
&I-\beta(e_{-i,-j}-e_{j,i}), \\
&I+\beta(e_{i,-j}-e_{j,-i}), \\
&I+\beta(e_{-i,j}-e_{-j,i})  \mid 0 < i < j, \beta \in \mathbb{F}_3\},    
\end{align*} 
where $e_{i,j}$ is the matrix which has $1$ in the $(i, j)$th entry and $0$ in other places. Then $\mathrm{\Omega^{+}}(2d,3)=\langle L \rangle$ (see e.g., \cite[Page~185]{carter-book}). Note that as $|Z| = 2$, we have that for each $U \in L$, $\ord(U)=3$ and $(|Z|,\ord(U))=1$.

Define $\lambda_{g_{i}} \in \aut(\POmegaPlus(2d,3))$ as $\lambda_{g_{i}}= \mathrm{Iso} \circ C_{g_{i}} \circ \mathrm{Iso}^{-1}$. Thus $A \cong \langle \lambda_{g_{1}}, \ldots, \lambda_{g_{t}} \rangle \leq \aut(\POmegaPlus(2d,3))$. We apply \Cref{construct-alpha-POmega} to each $\lambda_{g_i}$ to get $\alpha_{g_i} \in \aut(\Omega^+(2d,3))$ such that $\lambda_{g_i}=\bar{\alpha}_{g_i}$. The subgroup $\langle \alpha_{g_1},\ldots, \alpha_{g_t}\rangle$ is the natural embedding of $A$ in $\aut(\Omega^+(2d,3))$.

Note that now it is enough to check if $\alpha_{g_i}$ is in ${\rm{PO}}^+(2d,3)$ for each $i \in [t]$. Equivalently, we want to check if $\alpha_{g_i}$ is induced by an element of ${\rm{PO}}^+(2d,3)$. Fix $i$ and consider the following system of $|L|d^2$ equation for each $\alpha_{g_i}$:
\[
F_{{g_i}}U_j=\alpha_{g_{i}}(U_j)F_{{g_i}}  \text{ for each } 1 \leq j \leq |L|,
\]
where the unknowns are the entries of $F_{{g_i}}$. This system has $d^2|L|$ equations on $d^2$ unknowns. We can solve this system of  equations in $\textsf{NC}^2$ \cite{MulmuleyRank}. By \Cref{lemma-module-equ-cond-POmega}, a non-zero solution to the system always exists. Moreover, if $F_{{g_i}}$ is solution to the above system, then by \Cref{thm-schur}, $F_{{g_i}}\in {\rm{GL}}(2d,3)$.  Thus, the automorphism $\alpha_{g_i}$ is induced by the projective image of $F_{{g_i}} \in {\rm{GL}}(2d,
3)$. Further, $\alpha_{g_{i}} \in \mathrm{Aut(\Omega^{+}}(2d,3))$ is in ${\rm{PO}}^{+}(2d,3)$ if and only if the matrix $F_{g_{i}}$, satisfy $F_{g_{i}}X F_{g_{i}}^{t} =X$ i.e., $F_{g_{i}} \in \mathrm{O^{+}}(2d,3)$. 

For each $\alpha_{g_{i}}$ we find the matrix $F_{g_{i}}$ in $\textsf{NC}^2$ as described above. Also, note that $|L|$ is polynomially bounded. For each $F_{g_{i}}$ we check if $F_{g_{i}}X F_{g_{i}}^{T} =X$. Therefore we can determine if $A \not \leq {\rm{PO}}^{+}(2d,3)$, and if $3 \nmid |A/S_1|$ in $\textsf{NC}$. If $A \not \leq {\rm{PO}}^{+}(2d,3)$ and $3 \nmid |A/S_1|$ then $\mu(G,N)= \frac{(3^d-1)(3^{d-1}+1)}{2}$, otherwise $\mu(G,N)=\mu(S_1)$ (see e.g., \Cref{MD-AlmostSimple}).

\item \textbf{Computing a Minimal Faithful Permutation Representation of $A$.} 

If $\mu(A)=\mu(S)$ then a minimal faithful permutation representation of $A$ on the set of $1$-dimensional anisotropic subspaces of $V$ is induced by a minimal faithful permutation representation of $S$ on the same set. 

If $\mu(A) \neq \mu(S)$ then we proceed as follows. It follows from \cite[Proposition~1]{ref4MV} that ${\rm{PO}}^+(2d,3)$ admits a minimal faithful permutation representation on the set of $1$-dimensional subspaces isotropic with respect to 
$X$ (see \Cref{subsec: simplegroup}). As we are given $V$ explicitly, we may construct this set of $1$-dimensional subspaces of $V$ isotropic with respect to 
$X$ in $\textsf{NC}$. Then for each generator $\alpha_{g_{i}}$ of $A$, the induced matrix $F_{\alpha_{g_{i}}}$ acts on the set of $1$-dimensional subspaces of $V$ isotropic with respect to a nondegenerate quadratic form $X$ naturally via matrix multiplication. \qedhere
\end{itemize}
\end{proof}

\section{\texorpdfstring{$S_1 \cong \mathcal{L}(\mathbb{F}_{q})$, $\mathcal{L}(\mathbb{F}_{q}) \in \{\G(3^e), \F(2^e), \Esix(q)\}$}.}\label{sec:exceptional-groups}

In this section, we handle the remaining simple groups (simple groups of Lie type from rows 14, 15, 16 in \Cref{Table-MD-AlmostSimple}). 

\begin{proposition} \label{prop:LieConstructive1}
Let $G, K \trianglelefteq \Sym(n)$, with $K \trianglelefteq G$. Let $\textbf{G} := G/K$ be Fitting-free with minimal normal subgroup $N = S_{1} \times \cdots \times S_{\ell}$, where the $S_{i}$'s are isomorphic non-abelian simple groups. Suppose that each $S_{i}$ is isomorphic to $\mathcal{L}(\mathbb{F}_{q}) \in \{ \G(3^e), \F(2^e), \Esix(q) \}$. Additionally, suppose that if $\mathcal{L}(\mathbb{F}_{q}) \cong \G(3^e)$, then $e > 1$. Furthermore, suppose that we are given the following:
\begin{itemize}
\item Generators, as well as the multiplication table, for each $S_i$.

\item A minimal faithful permutation representation $\text{Iso} : S_1 \to S_1$ detailing the action of $S_1$ on the cosets of $S_1/P$, where $P$ is a maximal subgroup of $S_1$. In particular, we are given a $2$-element generating sequence of $S_1$ satisfying the Chevalley presentation for $\mathcal{L}(\mathbb{F}_{q})$.

\item $N_{\textbf{G}}(S_1) = \langle g_1, \ldots, g_t \rangle$, as well as $C_{\textbf{G}}(S_1)$.

\item For each $i \in [t]$, let $C_{g_{i}} : S_1 \to S_1$ be given by $C_{g_{i}}(s) = g_{i}^{-1}sg_{i}$. Suppose that we are given $A = \langle C_{g_{1}}, \ldots, C_{g_{t}} \rangle$.
\end{itemize}

Then we can compute $\mu(\textbf{G}, N) := \mu(A)$, as well as a minimal faithful permutation representation of $A$, in $\textsf{NC}$.
\end{proposition}

Our approach deviates slightly from what we have done in the previous sections. The reason is that the exceptional groups have small orders, allowing us to compute their Cayley tables. However, we also do not have the matrix group structure as in earlier cases, which makes the analysis more intricate.

\begin{proof}
We will describe the details for the case when $S_1 \cong \G(3^e)$. The other cases follow similarly. Note that if $S_1 \cong \G(3^e)$, our aim in this case is to determine whether $A \not \leq \Gamma \G(3^e)$ in the 3rd column of \Cref{Table-MD-AlmostSimple}. We have $A, \Gamma \G(3^e) \leq \aut(\G(3^e))$, and $\aut(\G(3^e)) \setminus \Gamma \G(3^e)$ are exactly those automorphisms which contain a  graph automorphism in their decompositions. Thus, it is enough to check if each generator of $A$ is free of graph automorphism.

We now describe the computational details. 

\begin{itemize}
\item \textbf{Determining $\mu(A)$.} 
Define $\lambda_{g_{i}} \in \aut(\mathcal{L}(\mathbb{F}))$ as $\lambda_{g_{i}}= \mathrm{Iso} \circ C_{g_{i}} \circ \mathrm{Iso}^{-1}$. Since $\aut(S_1)= \aut(\G(3^e))$ we have $A=\langle \lambda_{g_{1}},\ldots, \lambda_{g_{t}}\rangle \leq  \aut(\G(3^e))$ using an isomorphism, Iso.

For each $\lambda_{g_{i}}$, define $\alpha_{g_{i}}^{t',t''}=\lambda_{g_{i}}\mathbbm{f}_{\sigma_0}^{-t'}\mathbbm{g}^{-t''}$ for $0 \leq t' < e$ and $t'' \in \{0,1\}$. Now we want to check the type of $\alpha_{g_{i}}^{t',t''}$ which will be used to identify the type of $\lambda_{g_{i}}$. For each $i$, and each $t',t''$, we check if $\alpha_{g_{i}}^{t',t''}$ is either an inner automorphism or a diagonal automorphism by checking its value at all the elements of a group. This can be done since we have a Cayley table of $\G(3^e)$. Thus we can identify the type $\alpha_{g_{i}}^{t',t''}$. For a fix $i$, there are $2e$ possible $\alpha_{g_{i}}^{t',t''}$. Observe that out of these $2e$ possibilities, there is exactly one choice of $t'$ and $t''$ such that $\alpha_{g_{i}}^{t',t''}$ is either an inner or a diagonal automorphism. 

Suppose we identify that $\alpha_{g_{i}}^{t',t''}$ is an inner or a diagonal automorphism for some $t',t''$. If $t''=0$ then we get  $\lambda_{g_{i}}=\alpha_{g_{i}}^{t',t''}  \mathbbm{f}_{\sigma_0}^{t'}$ and $\lambda_{g_{i}} \in \Gamma \G(3^e)$. Thus for each $\lambda_{g_{i}}$ we can test if $\lambda_{g_{i}} \in \Gamma \G(3^e)$. Hence, we can identify if $A \not \leq \Gamma \G(3^e)$. If $A  \leq \Gamma \G(3^e)$ then $\mu(G,N)=\mu(S_1)$ otherwise $\mu(G,N)=2\mu(S_1)$ (see e.g., Lemma 2.4, \cite{CannonHoltUnger}).

Note that we can write each Chevalley generator explicitly, and we can also check if the automorphism is inner or a diagonal in $\textsf{NC}$.

\item \textbf{Computing a Minimal Faithful Permutation Representation of $A$.}  Note that the minimal faithful permutation representation of a simple group is always transitive \cite{thesis-chamberlain}. Therefore, to compute one such for $\G(3^e)$, it is enough to find a maximal subgroup such that the action on its cosets gives a minimal faithful permutation representation of $\G(3^e)$. Since all the maximal subgroups can be generated by at most four elements \cite{BurnessLiebeckShalev}, we can compute all such maximal subgroups for which the action of $\G(3^e)$ on their cosets is faithful. By \cite[Section 2]{ref5V}, there are up to conjugation, two such maximal subgroups, say $P_1$ and $P_2$ of index $\mu(\G(3^e))$. Without loss of generality, let $P_1$ denote the subgroup $P$ given in the statement of the proposition. One can verify that the graph automorphism $\mathbbm{g}$ sends $P_1$ to $P_2$ \cite{ref5V}. We can find both $P_1$ and $P_2$ in $\textsf{NC}$ as $|\G(3^e)| \leq n^9$. Moreover, we can also list $\{h_1 P_1,\dots, h_{\mu(S_1)} P_1\}$, a left coset transversal of $P_1$, in $\textsf{NC}$ (\Cref{lem:List}).

The generators $\lambda_{g_i}$ of $A$ are automorphisms of $\G(3^e)$, each $\lambda_{g_i}$ is composition of inner/diagonal ($h(\chi)$), field ($\mathbbm{f}_{\sigma_0}$), or graph automorphism ($\mathbbm{g}$) of $\G(3^e)$. Each of these automorphism is defined explicitly in \Cref{subsec: simplegroup}. As discussed before, we can write each Chevalley generators explicitly, and we can also check if the automorphism is inner or a diagonal in polynomial-time in the setting of quotients by and in $\textsf{NC}$ in the setting of permutation groups. It is then enough to define the action of each of these automorphisms of $\G(3^e)$ to get a minimal faithful permutation representation of $\G(3^e)$.

We first turn to constructing a permutation representation of $ \mathbbm{g}$, the graph automorphism. Recall that $\mu(A) = 2\mu(S)$ if and only if some generator of $A$ contains the graph automorphism. In particular, we require $2\mu(S)$ points to construct a minimal faithful permutation representation of $ \mathbbm{g}$. For this, we consider the set $\{h_1 P_1,\dots, h_{\mu(S_1)} P_1,\mathbbm{g}(h_1)P_2,\dots,\mathbbm{g}(h_{\mu(S_1)}) P_2\}$. As $\mathbbm{g}$ is a graph automorphism which sends $P_1$ to $P_2$, each cosets $\mathbbm{g}(h_{i}) P_2$ are distinct for $i=1,\dots,\mu(S_1)$. The action of $\mathbbm{g}$ on $\{h_1 P_1,\dots, h_{\mu(S_1)} P_1,\mathbbm{g}(h_1)P_2,\dots,\mathbbm{g}(h_{\mu(S_1)}) \}$ is defined as $\mathbbm{g}(h_iP_1)=\mathbbm{g}(h_i) P_2$ and $\mathbbm{g}(\mathbbm{g}(h_i) P_2)=h_i P_1$ for $i=1,\dots,\mu(S_1)$.

If $\alpha \in \Aut(G)$ is any of field, inner, or diagonal  automorphism, then we can define the induced action $\alpha(gP_{1}) = \alpha(g)P_{1}$ and $\alpha(gP_{2}) = \alpha(g)P_{2}$.

If $\mu(A)=\mu(S)$ then we restrict the permutation representation to the points $\{h_1 P_1,\dots, h_{\mu(S_1)} P_1\}$. Since $A$ is generated by the elements of $S_1$, together with diagonal $(h(\chi))$ field $(\mathbbm{f}_{\sigma_0})$, and graph $(\mathbbm{g})$ automorphisms, we have constructed a minimal faithful permutation representation of $A$. 
\end{itemize}

The case when $S_1$ is isomorphic to either $\F(2^e)$ or $ \Esix(q)$ is handled similarly as above -- we use in place of $\G(3^e)$ and $\Gamma\G(3^e)$ either: $\F(2^e)$ and $\Gamma\F(2^e)$; or $\Esix(q)$ and $\Gamma\Esix(q)$. The proof then goes through \emph{mutatis mutandis}.
\end{proof}

\section{Handling the Remaining Simple Group Cases} \label{sec:RemainingSimpleGroups}

In this section, we handle the remaining case from \Cref{Table-MD-AlmostSimple}, as well as all the simple groups when $\mu(A)=\mu(S)$. We show that in each case, one can construct a minimal faithful permutation representation of $S$ and $A$ efficiently. 

\begin{proposition}
Let $G, K \leq \text{Sym}(n)$, with $K \trianglelefteq G$. Let $\textbf{G} := G/K$ be Fitting-free, with minimal normal subgroup $N = S_{1} \times \cdots \times S_{\ell}$, where the $S_{i}$'s are isomorphic non-abelian simple groups, isomorphic to some non-abelian simple group from rows 1-9 of Table~\ref{Table-MD-AlmostSimple}. Furthermore, suppose that we are given the following:
\begin{itemize}
\item Generators for each of the $S_{i}$'s.
\item A minimal faithful permutation representation $\text{Iso} : S_{1} \to S_{1}$.
\end{itemize}

\noindent Then we can compute $\mu(\textbf{G},N) := \mu(A)$, as well as a minimal faithful permutation representation of $A$, in $\textsf{NC}$.
\end{proposition}

\begin{proof}
Note that these groups mentioned in rows 1-9 of \Cref{Table-MD-AlmostSimple} are of constant size and so we can easily construct a minimum permutation representation of $S_1$ and $A$ in $\textsf{NC}$, as follows. We consider in parallel, each $d \leq |S_{1}|$. For each such $d$, list all possible maps $\varphi : S_{1} \to \text{Sym}(d)$ and test whether each such map is an injective homomorphism. We then take such a $\varphi$ for the smallest such $d$ as our minimal faithful permutation representation. Similarly, we also compute a minimal faithful permutation representation of $A$ as well. 
\end{proof}

From now on, we only consider simple groups $S$ such that $\mu(A)=\mu(S_1)$, which are not part of \Cref{Table-MD-AlmostSimple}.

\begin{proposition}
Let $G, K \leq \text{Sym}(n)$, with $K \trianglelefteq G$. Let $\textbf{G} := G/K$ be Fitting-free, with minimal normal subgroup $N = S_{1} \times \cdots \times S_{\ell}$, where the $S_{i}$'s are isomorphic non-abelian simple groups, isomorphic to $\text{Alt}(m)$ for some $m \geq 6$. Furthermore, suppose that we are given the following:
\begin{itemize}
\item Generators for each of the $S_{i}$'s.
\item A minimal faithful permutation representation $\text{Iso} : S_{1} \to S_{1}$.

\item $N_{\textbf{G}}(S_1) = \langle g_1, \ldots, g_t \rangle$, as well as $C_{\textbf{G}}(S_1)$.
\item For each $i \in [t]$, let $C_{g_{i}} : S_1 \to S_1$ be given by $C_{g_i}(s) = g_{i}^{-1}sg_{i}$. Suppose that we are given $A = \langle C_{g_{i}}, \ldots, C_{g_{t}} \rangle$.
\end{itemize}

\noindent Then we can compute $\mu(\textbf{G},N) := \mu(A)$, as well as a minimal faithful permutation representation of $A$, in $\textsf{NC}$.
\end{proposition}

\begin{proof}
Note that $\aut(S_1)=\text{Sym}(m)$ and $A$ can be either $\text{Alt}(m)$ or $\text{Sym}(m)$. We now turn to determining whether $A = \text{Sym}(n)$. Recall that $N_{G}(S_{1})/C_{G}(S_{1}) \cong A$. Thus, $A = \text{Sym}(n)$ if and only if $|N_{G}(S_{1})| / |C_{G}(S_{1})| = n!$. By \Cref{PermutationGroupsNC}(a), we can compute $|N_{G}(S_{1})|$ and $|C_{G}(S_{1})|$ in $\textsf{NC}$. Thus, as $N_{G}(S_{1})/C_{G}(S_{1}) \cong A$, we can compute $|A|$ in $\textsf{NC}$. Now $A \cong \text{Alt}(m)$ if and only if $|A| = m!/2$. In this case, we may use Theorem~\ref{thm:ConstructiveRecognition} to obtain a minimal faithful permutation representation of $A$. 

Otherwise, $A \cong \Sym(2) \ltimes S_{1}$. As $\text{Iso}$ is a minimal faithful permutation representation of $S_{1} \cong \text{Alt}(m)$, we may build a corresponding permutation for the $\Sym(2)$ factor using standard constructions. Thus, in $\textsf{NC}$, we can construct a minimal faithful permutation representation of $A \cong \Sym(m)$. The result now follows.
\end{proof}

\begin{proposition} \label{prop:RemainingClassical}
Let $G, K \leq \text{Sym}(n)$, with $K \trianglelefteq G$. Let $\textbf{G} := G/K$ be Fitting-free, with minimal normal subgroup $N = S_{1} \times \cdots \times S_{\ell}$, where the $S_{i}$'s are isomorphic non-abelian simple groups, isomorphic to a classical simple group of dimension $d$ over $\mathbb{F}_{q}$ not listed in Table~\ref{Table-MD-AlmostSimple}. Suppose that $|S_{i}| \geq n^9$. Furthermore, suppose that we are given the following:
\begin{itemize}
\item Generators for each of the $S_{i}$'s.
\item A minimal faithful permutation representation $\text{Iso} : S_{1} \to S_{1}$, detailing the action of $S_{1}$ on some orbit $\overline{V}$ of $1$-dimensional subspaces of $V=\mathbb{F}_{q}^{d}$ .

\item A basis $\beta$ for $V$.

\item $N_{\textbf{G}}(S_1) = \langle g_1, \ldots, g_t \rangle$, as well as $C_{\textbf{G}}(S_1)$.
\item For each $i \in [t]$, let $C_{g_{i}} : S_1 \to S_1$ be given by $C_{g_i}(s) = g_{i}^{-1}sg_{i}$. Suppose that we are given $A = \langle C_{g_{i}}, \ldots, C_{g_{t}} \rangle$.
\end{itemize}

\noindent Then we can compute $\mu(\textbf{G},N) := \mu(A)$, as well as a minimal faithful permutation representation of $A$, in $\textsf{NC}$.
\end{proposition}

\begin{proof} 
For each $i \in [t]$, define $\lambda_{i} \in \Aut(S_{1})$ by $\lambda_{g_{i}} := \text{Iso} \circ C_{g_{i}} \circ \text{Iso}^{-1}$. Thus, $A$ is embedded as the subgroup $\langle \lambda_{g_{1}}, \ldots, \lambda_{g_{t}} \rangle$ in $\Aut(S_{1})$. Since $S_1$ does not appear in \Cref{Table-MD-AlmostSimple}, each $\lambda_{g_{i}}$ ($i \in [t]$) is free of graph automorphisms \cite{CannonHoltUnger}. In particular, each $\lambda_{g_{i}}$ is a semilinear transformation on $V$, and so is (abstractly) a permutation of $V$ (though need not be represented as a permutation of $V$ at this stage). As we have a basis $\beta$ of $V$, we may apply \cite[Lemma~6.5(iii)]{KantorLuksMark} to obtain $M = (a_{ij}) \leq V$ and $\sigma \in \Aut(\mathbb{F}_{q})$ such that the pair $(M, \sigma)$ induce the semilinear transformation $\tau : \sum_{i=1}^{d} c_{i}e_{i} \mapsto \sum_{i=1}^{d} c_{i}^{\sigma} a_{ij} e_{j}$. By applying $\tau$ to $V$, we may construct a permutation representation $\psi_{i}$ of $\lambda_{g_{i}}$ as follows. For a $1$-space $[v]$ of $V$, $\psi_{i}$ maps $[v] \mapsto [\tau v]$. Thus, we have constructed a minimal faithful permutation representation of $A$, as desired. 
\end{proof}

\begin{proposition} \label{prop:LieConstructive2}
Let $G, K \trianglelefteq \Sym(n)$, with $K \trianglelefteq G$. Let $\textbf{G} := G/K$ be Fitting-free with minimal normal subgroup $N = S_{1} \times \cdots \times S_{\ell}$, where the $S_{i}$'s are isomorphic non-abelian simple groups. Suppose that each $S_{i}$ is an exceptional group of Lie type or a classical simple group of order less than $n^9$, that is not listed in Table~\ref{Table-MD-AlmostSimple}. Furthermore, suppose that we are given the following:
\begin{itemize}
\item Generators, as well as the multiplication table, for each $S_i$.

\item A minimal faithful permutation representation $\text{Iso} : S_1 \to S_1$ detailing the action of $S_1$ on the cosets of $S_1/P$, where $P$ is a maximal subgroup of $S_1$. 

\item $N_{\textbf{G}}(S_1) = \langle g_1, \ldots, g_t \rangle$, as well as $C_{\textbf{G}}(S_1)$.

\item For each $i \in [t]$, let $C_{g_{i}} : S_1 \to S_1$ be given by $C_{g_{i}}(s) = g_{i}^{-1}sg_{i}$. Suppose that we are given $A = \langle C_{g_{1}}, \ldots, C_{g_{t}} \rangle$.
\end{itemize}

Then we can compute $\mu(\textbf{G}, N) := \mu(A)$, as well as a minimal faithful permutation representation of $A$, in $\textsf{NC}$.
\end{proposition}

\begin{proof}
Define $\lambda_{g_{i}} \in \aut(S_1)$ as $\lambda_{g_{i}}= \mathrm{Iso} \circ C_{g_{i}} \circ \mathrm{Iso}^{-1}$. This provides an embedding of $A$ into $\aut(S_1)$. For each $\lambda_{g_{i}}$ ($i \in [t]$), we construct a permutation representation $\psi_{i}$ by mapping the coset $gP \mapsto (\lambda_{g_{i}}(g))P$. Since $S_1$ does not appear in \Cref{Table-MD-AlmostSimple}, we have $\mu(A)=\mu(S_1)$. In particular, $A$ is free of the graph automorphisms. Thus, we have constructed a minimal faithful permutation representation of $A$, as desired.
\end{proof}

\section{Conclusion}

We investigated the computational complexity of \algprobm{Min-Per-Deg} for Fitting-free groups. When the groups are given as quotients of two permutation groups, we established polynomial-time bounds for computing both the minimal faithful permutation degree $\mu(G)$, as well as a permutation representation $\varphi : G \to \Sym(\mu(G))$. Furthermore, in the setting of permutation groups, we established an upper bound of $\textsf{NC}$ for computing $\mu(G)$, and $\textsf{RNC} $ for constructing a minimal faithful permutation representation. In the process, we exhibited an upper bound of $\textsf{NC}$ for computing the socle of a Fitting-free group, in the setting of permutation groups. Our work leaves open several problems. The most natural open problem is the following.

\begin{problem} \label{prob:p1}
In the setting of quotients, does \algprobm{Min-Per-Deg} for Fitting-free groups belong to $\textsf{NC}$?
\end{problem}

Resolving Problem~\ref{prob:p1} will likely require novel parallel solutions for much of the polynomial-time framework for quotients developed by Kantor and Luks \cite{KantorLuksQuotients}. We highlight key obstacles next.

\begin{problem}
Let $G, K \leq \Sym(n)$ with $K \trianglelefteq G$, and let $\textbf{G} := G/K$. Let $\textbf{N} \trianglelefteq \textbf{G}$. Compute the centralizer $C_{\textbf{G}}(\textbf{N})$ in $\textsf{NC}$.   
\end{problem}

In order to obtain a polynomial-time solution for computing $C_{\textbf{G}}(\textbf{N})$ in the setting of quotients, Kantor and Luks \cite{KantorLuksQuotients} reduce to computing the intersection $G \cap P$ between an arbitrary permutation group $G$ and a $p$-group $P$. Thus, we ask the following.

\begin{problem} \label{prob:Cores}
Let $G, P \leq \Sym(n)$ with $P$ a $p$-group. Can we compute $G \cap P$ in $\textsf{NC}$?
\end{problem}

\begin{problem}
Let $G, K \leq \Sym(n)$ with $K \trianglelefteq G$, and let $\textbf{G} := G/K$. Suppose that $\textbf{G}$ is Fitting-free. Compute $\Soc(\textbf{G})$ in $\textsf{NC}$.
\end{problem}

More generally, we ask about computing $\Soc(\textbf{G})$ in $\textsf{NC}$ when $\textbf{G}$ is \emph{not} assumed to be Fitting-free. This is open, even in the setting of permutation groups ($K = 1$).

Theorem~\ref{thm:MainFittingFree}(a) formalizes that much of the framework established by Cannon and Holt \cite{CH03} is polynomial-time computable. In order to strengthen these results to the setting of $\textsf{NC}$, it is necessary but not sufficient to resolve Problem~\ref{prob:p1}. We also need to compute the solvable radical.

\begin{problem} \label{prob:Radical}
Let $G, K \leq \Sym(n)$ with $K \trianglelefteq G$, and let $\textbf{G} := G/K$. Compute $\rad(\textbf{G})$ in $\textsf{NC}$.
\end{problem}

Problem~\ref{prob:Radical} is open, even in the setting of permutation groups ($K = 1$).

Finally, it remains open to derandomize Theorem~\ref{thm:MainFittingFree}(b). It would suffice to resolve the following problem.

\begin{problem}[{\cite[Problem~5]{BabaiLuksSeress}}]
Let $H \leq G \leq \Sym(n)$ be permutation groups. Suppose that $[G : H] \in \poly(n)$. Can a transversal for the cosets of $H$ in $G$ be constructed in $\textsf{NC}$?
\end{problem}

\section*{Acknowledgments}
We wish to thank Peter Brooksbank, Josh Grochow, Takunari Miyazaki, Eamonn O'Brien, and Sheila Sundaram for many helpful discussions. ML completed parts of this work at Tensors Algebra-Geometry-Applications (TAGA) 2024. ML thanks Elina Robeva, Christopher Voll, and James B. Wilson for organizing the conference. ML was partially supported by CRC 358 Integral Structures in Geometry and Number Theory at Bielefeld and Paderborn, Germany; the Department of Mathematics at Colorado State University; James B. Wilson's NSF grant DMS-2319370; and the Computer Science department at the College of Charleston. DT is supported by JSPS KAKENHI grant No. JP24H00071.

\bibliographystyle{alphaurl}
\bibliography{references}

\end{document}